\documentclass[journal]{IEEEtran}
\makeatletter
\def\ps@headings{%
\def\@oddhead{\mbox{}\scriptsize\rightmark \hfil \thepage}%
\def\@evenhead{\scriptsize\thepage \hfil \leftmark\mbox{}}%
\def\@oddfoot{}%
\def\@evenfoot{}}

\makeatother 
\IEEEoverridecommandlockouts
\usepackage{bbm}
\usepackage{amsfonts}
\usepackage[dvips]{graphicx}
\usepackage{times}
\usepackage{cite}
\usepackage{amsmath}
\usepackage{array}
\usepackage{amssymb}
\usepackage{stfloats}
\usepackage{slashbox}
\usepackage{graphicx}
\usepackage{footnote}
\usepackage{booktabs}
\usepackage{array}
\usepackage{algorithm}
\usepackage{subeqnarray}
\usepackage{cases}
\usepackage{color}
\usepackage{hyperref}
\usepackage{epstopdf}
\usepackage{algpseudocode}
\usepackage{subcaption}
\usepackage{multirow}
\usepackage{adjustbox}
\usepackage[labelformat=simple]{subcaption}
\usepackage[tableposition=top]{caption}
\newcommand{\tabincell}[2]{\begin{tabular}{@{}#1@{}}#2\end{tabular}}

\newcommand*{\QEDA}{\hfill\ensuremath{\blacksquare}}%
\newtheorem{proof}{Proof}
\newtheorem{theorem}{Theorem}

\newtheorem{proposition}{Proposition}
\newtheorem{lemma}{Lemma}
\newtheorem{definition}{Definition}
\newtheorem{problem}{Problem}

\begin{document}

\title{Stochastic Content-Centric  Multicast Scheduling for Cache-Enabled Heterogeneous Cellular Networks}

\author{Bo~Zhou, Ying~Cui,~\IEEEmembership{Member,~IEEE}, and Meixia~Tao,~\IEEEmembership{Senior~Member,~IEEE}
\thanks{
This paper was presented in part at  ACM CoNEXT 2015 Workshop on Content Caching and Delivery in Wireless Networks (CCDWN), Heidelberg, Germany, December, 2015 \cite{ccdwn}.

B.~Zhou, Y.~Cui and M.~Tao are with the Department of Electronic Engineering at Shanghai Jiao Tong University, Shanghai,
200240, P. R. China. Email: \{b.zhou, cuiying, mxtao\}@sjtu.edu.cn.}
}

\maketitle

\begin{abstract}
Caching at small base stations (SBSs) has demonstrated significant benefits in alleviating the backhaul requirement in heterogeneous cellular networks (HetNets). While many existing works focus on what contents to cache at each SBS, an equally important problem is what contents to deliver so as to satisfy dynamic user demands given the cache status.
In this paper, we study optimal content delivery in cache-enabled HetNets by taking into account the inherent multicast capability of wireless medium.
We consider stochastic content multicast scheduling to jointly minimize the average network delay and power costs under a multiple access constraint. We establish a content-centric request queue model and formulate this stochastic optimization problem as an infinite horizon average cost Markov decision process (MDP).
By using  \emph{relative value iteration} and special properties of the request queue dynamics, we characterize some properties of the value function of the MDP.
Based on these properties, we show that the optimal multicast scheduling policy is of threshold type.
Then, we propose a structure-aware optimal algorithm to obtain the optimal policy.
We also propose a low-complexity suboptimal policy, which possesses similar structural properties to the optimal policy, and develop a low-complexity algorithm to obtain this policy.
\end{abstract}


\section{Introduction}\label{sec:introduction}
The rapid proliferation of smart mobile devices has trigged an unprecedented growth of the global mobile data traffic, which is expected to reach 24.3 exabytes per month by 2019 \cite{Cisco}.
One promising approach to meet the dramatic traffic growth is to deploy small base stations (SBSs) together with traditional macro base stations (MBSs) in a heterogeneous network paradigm \cite{6211486}.
Such a heterogeneous cellular network (HetNet) provides short-range localized communications by bringing base stations (BSs) closer to users, and hence increases the area spectral efficiency and network capacity.
However, the main drawback of this approach is the requirement of expensive high-speed backhaul links for connecting all  SBSs to the core network. The backhaul capacity requirement can be enormously high during peak traffic hours.

Recently, caching at BSs has been proposed as an effective way to alleviate the backhaul capacity requirement and improve the user-perceived quality of experience in wireless networks \cite{molisch2014caching,Debbah,paschos2016wireless,7282567}.
Caching has received significant attention in the literature \cite{femto,7218465,6665021,Bastug2015,7194828}.
Specifically, in \cite{femto}, the authors introduce the concept of \emph{FemtoCaching} and study content placement at  SBSs to minimize the average content access delay.
In \cite{7218465}, the authors consider  joint request routing and content caching for HetNets and propose approximate algorithms to minimize the content access delay.
In \cite{6665021}, the authors study the joint optimization of power and cache control for video streaming in multi-cell multi-user MIMO systems.
References \cite{Bastug2015} and \cite{7194828} study outage probability and average delivery rate of cache-enabled HetNets for given caching strategies.
However, \cite{femto,7218465,6665021,Bastug2015,7194828} consider point-to-point unicast transmission for cache-enabled wireless networks and can only help to reduce the backhaul burden without effectively relieving the ``on air'' congestion.
Furthermore, in \cite{femto,7218465,6665021,Bastug2015,7194828}, the inherent broadcast nature of wireless medium is not fully exploited, which is the major distinction of  wireless communications from wired communications.

Enabling multicasting at BSs is another efficient approach to deliver popular contents to multiple requesters concurrently. Wireless multicasting has been specified in 3GPP standards known as evolved Multimedia Broadcast Multicast Service (eMBMS) \cite{embms}.
In view of the benefits of caching and multicasting, joint design of the two promising techniques is expected to achieve superior performance for massive content delivery in wireless networks.
From an information-theoretic perspective, \cite{coded} proposes a novel coded caching scheme to minimize the peak traffic load for a single-cell network by utilizing multicast transmission and caches at users, and characterizes the memory-rate tradeoff.
Note that, \cite{coded} only considers the delay-insensitive services.
However, many content-centric applications, such as video steaming, are delay-sensitive, and it is critical to consider delay performance in cache-enabled content-centric wireless networks \cite{7249208,TWC16,ton,ISIT}.
In specific, the authors in \cite{7249208} study coded multicasting for inelastic services (with strict deadline) in a single-cell network, and propose a computationally efficient content delivery algorithm under a given coded caching scheme.
In \cite{TWC16}, the authors propose an approximate caching algorithm with performance guarantee and a heuristic caching algorithm, to reduce the service cost for inelastic services in a cache-enabled small-cell network under a given multicast transmission strategy.
In \cite{ton}, the authors consider multicasting for inelastic services in a cache-enabled multi-cell network, and propose joint throughput-optimal caching and scheduling algorithms to maximize the service rates of the inelastic services.
In our recent work \cite{ISIT}, we consider optimal multicast scheduling to jointly minimize the average delay and service costs of elastic services (delay-sensitive services but without strict deadlines) for a cache-enabled single-cell network.
However, for elastic services, it remains unknown how to design optimal multicast scheduling to jointly minimize  average delay and service costs for given cache placement  in cache-enabled HetNets.
The main challenge of determining  optimal multicast scheduling in HetNets stems from  the heterogeneous structure of the network (see a motivating example in Section~\ref{sec:example}).

In this paper, we consider a cache-enabled HetNet with one MBS, $N$ SBSs, $K$ users, and $M$ contents (with possibly different content sizes).
The SBS coverage areas are assumed to be disjoint.
Assume that the MBS and the SBSs are not allowed to operate concurrently, to avoid excessive interference, while the SBSs are allowed to operate at the same time without mutual interference.
This is referred to as the multiple access constraint.
Each SBS is equipped with a cache storing a certain number of contents, depending on the sizes of the cached contents and the cache size.
The MBS stores all contents in the network.
In each slot, each BS either schedules one cached content for multicasting to serve the pending requests from the users in its coverage area, or keeps idle, i.e., does not transmit any content.
We consider stochastic content multicast scheduling to jointly minimize the average network delay and power costs under the multiple access constraint.
We establish a content-centric request queue model and then formulate this stochastic multicast scheduling problem as an infinite horizon average cost Markov decision process (MDP)\cite{bertsekas}.
There are several technical challenges involved.

$\bullet$ \textbf{Optimality analysis:} The infinite horizon average cost MDP is well-known to be challenging due to the curse of dimensionality. Although dynamic programming provides a systematic approach for MDPs, there generally exist only numerical solutions.  These solutions do not typically offer many design insights and are usually impractical due to the curse of dimensionality \cite{bertsekas}. Thus, it is highly desirable to study the structural properties of the optimal policy.
There are several existing works on structural analysis for hybrid systems\cite{6572969,shifrin2015coded,altman2010forever}. However, the system models and queueing models in these works significantly differ from ours, and hence the approaches and results therein  cannot be straightforwardly extended to our work.
Specifically, our problem can be viewed as a problem of scheduling a single broadcast server (the MBS) or multiple multicast servers (the SBSs) to parallel (request) queues with general arrivals.
Existing works have only studied the problems of scheduling a single broadcast server to parallel queues \cite{ISIT,batch,cdc}.
Therefore, the structural analysis of the optimal multicast scheduling of a single broadcast server or multiple broadcast servers to parallel queues with general arrivals remains unknown and cannot be straightforwardly extended from the existing solutions.

$\bullet$ \textbf{Algorithm design:} Standard numerical algorithms such as value iteration and policy iteration to MDPs are usually computationally impractical for real systems.
To reduce the complexity, several structured optimal algorithms, which incorporate the structural properties into standard algorithms, are proposed\cite{djonin2007mimo,OR}. However, the curse of dimensionality still remains an issue.
On the other hand, the structural properties of the optimal policy may be one key reason for its good performance.
 Thus, it is highly desirable to further reduce the complexity of the structured optimal algorithms, while maintaining similar structural properties to the optimal policy.

By using \emph{relative value iteration algorithm} (RVIA) \cite[Chapter 4.3.1]{bertsekas} and special properties of the request queue dynamics, we characterize some properties of the value function of the MDP.
Based on these properties, we show that the optimal multicast scheduling policy, which is adaptive to the request
queue state, is of threshold type. This reveals the tradeoff between the average delay cost and the average power cost.
Then, we propose a structure-aware optimal algorithm by exploiting the structural properties of the optimal policy.
To further reduce the computational complexity, using approximate dynamic programming \cite{bertsekas}, we propose a low-complexity suboptimal policy, which possesses similar structural properties to the optimal policy, and develop a low-complexity algorithm to compute this policy.
Numerical examples verify the theoretical results and demonstrate the performance of the optimal and suboptimal solutions.

The rest of this paper is organized as follows. Section II introduces the network model. Section III provides the formulation of the stochastic multicast scheduling problem and the optimality equation. The structural properties of the optimal policy are presented in Section IV and a structure-aware optimal algorithm is proposed in Section V. In Section VI, we propose a low-complexity suboptimal solution. Numerical results are provided in Section VII. Finally, we conclude the paper and provide several future directions in Section VIII. The important notations used in this paper are summarized in Table~\ref{tablenotation}.
\begin{table}[!thbp]
\caption{List of important notations}\label{tablenotation}
\begin{tabular}{|c|c|}
\hline
$\mathcal{N}^+$ & set of all SBSs\\
\hline
$\mathcal{N}=\mathcal{N}^+\cup\{0\}$ & set of all BSs; BS $0$: MBS\\
\hline
$\mathcal{K}_0$   & set of users not covered by any SBS \\
\hline
$\mathcal{K}_n, n\in\mathcal{N}^+$   & set of users within coverage area of SBS $n$\\
\hline
$\mathcal{K}=\bigcup_{n\in\mathcal{N}}\mathcal{K}_n$ & set of all users\\
\hline
$\mathcal{M}$ & set of all contents\\
\hline
$\mathcal{M}_n, n\in\mathcal{N}$ & set of contents cached at BS $n$\\
\hline
$\mathcal{N}_m$ & set of SBSs that caching content $m$\\
\hline
$n,k,m,t$ & BS, user, content, slot index\\
\hline
$p(n,m)$ & \tabincell{c}{minimum transmission power required by BS $n$ to\\deliver content $m$ to all associated users within a slot} \\
\hline
$\mathbf{A}=(A_{n,m})$ & request arrival matrix\\
\hline
$\mathbf{Q}=(Q_{n,m})$ & request queue state matrix\\
\hline
$d(\mathbf{Q})$ & sum request queue length\\
\hline
$\mathbf{u}=(u_n)_{n\in\mathcal{N}}$ & multicast scheduling action\\
\hline
$\mu$ & stationary multicast scheduling policy\\
\hline
$V(\mathbf{Q})$ & value function\\
\hline
$J(\mathbf{Q},u)$ & state-action cost function\\
\hline
\end{tabular}
\centering
\end{table}
\section{Network Model}\label{sec:model}
\begin{figure}[!t]
\begin{centering}
\includegraphics[scale=0.72]{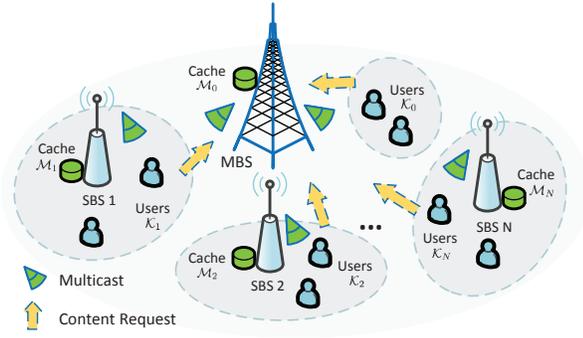}
 \caption{Cache-enabled  heterogeneous cellular network.}\label{fig:networkmodel}
\end{centering}
\end{figure}

Consider a cache-enabled HetNet with one MBS, $N$ SBSs, $K$ users, and $M$ contents, as illustrated in Fig.~\ref{fig:networkmodel}.
Let $\mathcal{N}\triangleq\{0,1,2,\cdots,N\}$ denote the set of all BSs, where BS $0$ refers to the MBS and BS $n=1,2,\cdots,N$ refers to SBS $n$.
Let $\mathcal{N}^+\triangleq\{1,2,\cdots,N\}$ denote the set of $N$ SBSs.
The SBS coverage areas are assumed to be disjoint.
Let $\mathcal{K}\triangleq\{1,2,\cdots,K\}$ denote the set of $K$ users in the network.
Let $\mathcal{K}_n\subseteq\mathcal{K}$ denote the set of users within the coverage area of SBS $n\in\mathcal{N}^+$.
All the users in $\mathcal{K}_n$ ($n\in\mathcal{N}^+$) can be served by the MBS and SBS $n$.
Let $\mathcal{K}_0\triangleq\mathcal{K}-\bigcup_{n\in\mathcal{N}^+}\mathcal{K}_n$ denote the set of users not covered by any SBS.
All the users in $\mathcal{K}_0$ can only be served by the MBS.
Note that, we do not distinguish the users in $\mathcal{K}_n$.
Let $\mathcal{M}\triangleq\{1,2,\cdots,M\}$ denote the set of $M$ contents (with possibly different content sizes) in the network.
Each BS is equipped with a cache storing a certain number of contents, depending on the cache size and the sizes of the cached contents.
Let $\mathcal{M}_n\subseteq\mathcal{M}$ denote the set of cached contents in BS $n\in\mathcal{N}$.
We assume $\mathcal{M}_0=\mathcal{M}$, i.e., the MBS stores all contents in the network.
Let $\mathcal{N}_m\triangleq\{n\in\mathcal{N}^+|m\in\mathcal{M}_n\}$ denote the set of SBSs caching content $m\in\mathcal{M}$.
We assume that the contents stored in the caches are given according to certain caching strategy (similar assumptions have been made in the literature, e.g. \cite{Bastug2015,7194828,7249208}) and consider multicast scheduling for a given caching design.
Notice that caching in general is in a much larger time-scale (e.g., on a weekly or monthly basis) while multicast scheduling is in a shorter time-scale \cite{femto,6665021,molisch2014caching}.
Consider time slots of unit length (without loss of generality) indexed by $t=1,2,\cdots$.
In the sequel, we first introduce the request arrival model, the service model, and the request queue model. Then, we provide a motivating example that highlights the challenges in designing the optimal multicast scheduling.


\subsection{Request Arrival Traffic}
In each slot, each user submits content requests to the MBS.
Let $A_{n,m}(t)\in\{0,1,\cdots\}$ denote  the number of the requests for content $m$  from all users in $\mathcal{K}_n$ which arrive during slot $t$, where $n\in\mathcal{N}$ and $m\in\mathcal{M}$.
Let $\mathbf{A}(t)\triangleq (A_{n,m}(t))_{n\in\mathcal{N},m\in\mathcal{M}}$ denote the request arrival matrix during slot $t$.
We assume that the request arrival processes $\{A_{n,m}(t)\}$ ($n\in\mathcal{N},m\in\mathcal{M}$) are mutually independent with respect to  $n$  and $m$; and $\{A_{n,m}(t)\}$ are i.i.d.  with respect to $t$ for all $n\in\mathcal{N}$ and $m\in\mathcal{M}$.
Note that, the content request arrivals are modeled according to the Independent Reference Model (IRM), which is a standard approach adopted in the literature \cite{femto,ton,rosensweig2013steady}.
The MBS maintains separate request queues for each BS $n\in\mathcal{N}$ and each cached content $m\in\mathcal{M}_n$.
The request queue model will be further illustrated in Section~\ref{sec:request_model}.

\subsection{Service Model}\label{sec:service_model}
We consider multicast service for content delivery in the network.
For ease of illustration, we assume that each BS can transmit at most one content in each slot. The analytical framework and results can be extended to the general case in which each BS can transmit multiple contents in one slot.
In each slot, each BS $n\in\mathcal{N}$ either schedules one cached content for multicasting to serve the pending requests from all users in its coverage area, or keeps idle (i.e., does not transmit any content).
We consider the Zero Download Delay (ZDD) assumption \cite{rosensweig2013steady}, i.e., all scheduled contents in one slot can be delivered to the users within the same slot.
Let $p(n,m)$ denote the minimum transmission power required by BS $n$ for successfully delivering one cached content $m$ to all users in its coverage area within a scheduling slot, where $n\in\mathcal{N}$ and $m\in\mathcal{M}_n$.
We set $p(n,0)=0$ for all $n\in\mathcal{N}$.
If BS $n$ multicasts content $m$ with transmission power $p(n,m)$, all pending requests for content $m$ from all users in the coverage area of BS $n$ are satisfied.
Let $u_n(t)\in\tilde{\mathcal{M}}_n\triangleq\mathcal{M}_n\cup\{0\}$ denote the scheduling action of BS $n\in\mathcal{N}$ at slot $t$, where
$u_n(t)\neq 0$ indicates that BS $n$ multicasts the cached content $u_n(t)$ with transmission power $p(n,u_n(t))$ at slot $t$, and
$u_n(t)=0$ indicates that BS $n$ does not transmit any content at slot $t$.
Let $\mathbf{u}(t)\triangleq(u_n(t))_{n\in\mathcal{N}}$ denote the multicast scheduling action in the network at slot $t$.

The MBS is assumed to operate at a much higher transmission power level than the SBSs, for providing full coverage of the network.
To avoid excessive interference, we therefore do not allow the MBS and the SBSs to operate concurrently.
On the other hand, since the SBSs are spatially separated and use much lower powers, we allow the SBSs to operate at the same timewithout mutual interference.
Mathematically, for all $t$, we require,
\begin{equation}\label{eqn:u_constraint}
u_0(t)\sum_{n\in\mathcal{N}^+}u_n(t)=0.
\end{equation}
We refer to \eqref{eqn:u_constraint} as the multiple access constraint.
Let $\mathbf{\mathcal{U}}\triangleq\{(u_n)_{n\in\mathcal{N}}|u_n\in\tilde{\mathcal{M}}_n~\forall n\in\mathcal{N}~\text{and}\allowbreak~u_0\sum_{n\in\mathcal{N}^+}u_n=0\}$ denote the feasible multicast scheduling action space.
The network power cost $p(\mathbf{u})$ associated with $\mathbf{u}\in\mathbf{\mathcal{U}}$ is given by
\begin{equation}
p(\mathbf{u})\triangleq \sum_{n\in\mathcal{N}} p(n,u_n).\label{eqn:power}
\end{equation}

\subsection{Request Queue Model} \label{sec:request_model}
As illustrated above, for each SBS $n\in\mathcal{N}^+$, the requests for cached content $m\in\mathcal{M}_n$ from all the users in $\mathcal{K}_n$  can be served by both the MBS and SBS $n$, while the requests for uncached content $m\in\mathcal{M}_0\setminus\mathcal{M}_n$ can only be served by the MBS.
On the other hand, the MBS can serve the requests for any content $m\in\mathcal{M}_0$ from the users in $\mathcal{K}_n$.
Therefore, the request queues maintained by the MBS are constructed as follows.
For each SBS $n\in\mathcal{N}^+$ and each cached content $m\in\mathcal{M}_n$, we construct a separate request queue, referred to as queue $(n,m)$, storing the requests for content $m$ from all the users in $\mathcal{K}_n$.
Let $Q_{n,m}(t)$ denote the length of queue $(n,m)$ at the beginning of slot $t$, where $n\in\mathcal{N}^+$ and $m\in\mathcal{M}_n$.
For the MBS $n=0$ and each content $m\in\mathcal{M}_0$, we also construct a separate request queue, referred to as queue $(0,m)$, storing the requests for content $m$ from all users in $\bigcup_{n\in\mathcal{N}^+\setminus\mathcal{N}_m}\mathcal{K}_n$ (the set of users covered by the SBSs where content $m$ is not cached) and $\mathcal{K}_0$.
Let $Q_{0,m}(t)$ denote the length of queue $(0,m)$ at the beginning of slot $t$, where $m\in\mathcal{M}_0$.
Note that these request queues can be implemented using counters and no data is contained in these queues.
For each queue $(n,m)$, let $N_{n,m}$ denote the upper limit of the corresponding counter.
For technical tractability, we assume that $N_{n,m}$ is finite (can be arbitrarily large).  This assumption is to guarantee that the request queue state space is finite, which would greatly simplify the mathematical arguments \cite{bertsekas}. Let $\mathcal{Q}_{n,m}\triangleq\{0,1,\cdots,N_{n,m}\}$ denote the request queue state space for queue $(n,m)$.
Let $\mathbf{Q}(t)\triangleq (Q_{n,m}(t))_{n\in\mathcal{N},m\in\mathcal{M}_n}\in\mathbf{\mathcal{Q}}$ denote the request queue state of the network at the beginning of slot $t$, where $\mathbf{\mathcal{Q}}\triangleq\prod_{n\in\mathcal{N}}\prod_{m\in\mathcal{M}_n}{\mathcal{Q}_{n,m}}$ denotes the request queue state space and the $\prod$ operation denotes the Cartesian product.

For each SBS $n\in\mathcal{N}^+$ and each cached content $m\in\mathcal{M}_n$, all the pending requests in queue $(n,m)$ are satisfied, if content $m$ is scheduled for multicasting by the MBS (i.e., $u_0(t)=m$) or by  SBS $n$ (i.e., $u_n(t)=m$) at slot $t$. Thus, for each $n\in\mathcal{N}^+$ and $m\in\mathcal{M}_n$, the request queue dynamics is as follows:
\begin{align}\label{eqn:queue1}
Q_{n,m}(t+1)=\min\{\mathbf{1}(u_0(t)\neq m~&\&~u_n(t)\neq m)Q_{n,m}(t)\nonumber\\
&+A_{n,m}(t),N_{n,m}\},
\end{align}
where  $\mathbf{1}(\cdot)$ denotes the indicator function.

For the MBS and each content $m\in\mathcal{M}_0$, all the pending requests in queue $(0,m)$ are satisfied, if content $m$ is scheduled for multicasting by the MBS at slot $t$ (i.e., $u_0(t)=m$).
Thus, for the MBS and each content $m\in\mathcal{M}_0$, the request queue dynamics is as follows:
\begin{align}\label{eqn:queue2}
Q_{0,m}(t+1)=\min\{\mathbf{1}(u_0(t)\neq m)Q_{0,m}(t)+\tilde{A}_{0,m}(t), N_{0,m}\},
\end{align}
where $\tilde{A}_{0,m}(t)\triangleq\sum_{n\in\mathcal{N}^+\setminus\mathcal{N}_m}A_{n,m}(t) +A_{0,m}(t)$   denotes the total number of requests for content $m$ from all users in $\bigcup_{n\in\mathcal{N}^+\setminus\mathcal{N}_m}\mathcal{K}_n$ and  $\mathcal{K}_0$ which arrive during slot $t$.
Note that each request is stored in only one queue.

\subsection{Motivating Example}\label{sec:example}
As illustrated in Fig.~\ref{fig:example}, consider a network  with 1 MBS, 2 SBSs ($\mathcal{N}^+=\{1,2\}$), 3 users ($\mathcal{K}=\{1,2,3\}$), and 3 contents ($\mathcal{M}=\{1,2,3\}$). We set $\mathcal{K}_1=\{1\}$, $\mathcal{K}_2=\{2\}$, $\mathcal{K}_0=\{3\}$, $\mathcal{M}_1=\{1,2\}$, $\mathcal{M}_2=\{2,3\}$, and  $\mathcal{M}_0=\{1,2,3\}$.
According to Section~\ref{sec:request_model}, the MBS maintains seven request queues, i.e., queues $(0,1)$, $(0,2)$, $(0,3)$, $(1,1)$, $(1,2)$, $(2,2)$, and $(2,3)$.
Our goal is to design the optimal multicast scheduling so as to jointly minimize the average network delay cost and power cost.
This involves two challenging and coupled tasks.
\begin{figure}[!t]
\begin{centering}
\includegraphics[scale=0.72]{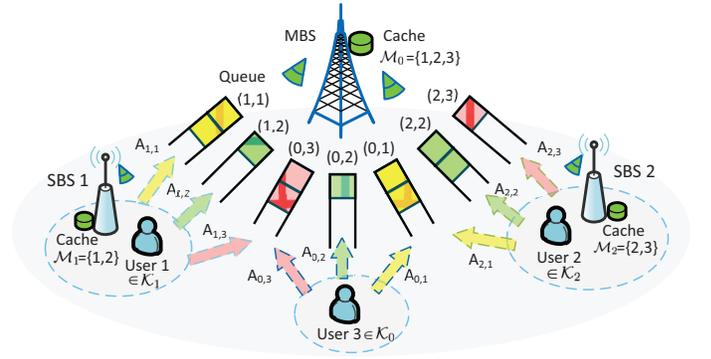}
 \caption{An example with 1 MBS, 2 SBSs, 3 users and 3 contents.}\label{fig:example}
\end{centering}
\end{figure}
First, at each time slot, it is not clear whether to operate the MBS or the SBSs.
If we schedule the MBS to multicast, then the pending requests for one content in the whole network can be satisfied with a higher power cost, e.g.,
clear queues $(0,2)$, $(1,2)$, and $(2,2)$ with power $p(0,2)$. If we schedule the SBSs to multicast, the pending requests for (possibly different) contents in different SBS coverage areas can be satisfied with a lower power cost, e.g., clear queues $(1,2)$ and $(2,3)$ with power $p(1,2)+p(2,3)$.

Second, at each time slot, if BS $n\in\mathcal{N}$ is allowed to operate, it is unknown whether to keep BS $n$ idle or not, and which content to schedule for multicasting if BS $n$ is not idle.
 Take SBS 1 as an example.
Suppose at certain time slot $t$ we have $Q_{1,1}(t)>Q_{1,2}(t)$ and $p(1,1)>p(1,2)$.
Scheduling Content 1 can satisfy  more requests with a higher power cost; scheduling Content 2 can satisfy fewer requests with a lower power cost; keeping SBS 1 idle consumes zero power cost.



rom this example, we can see that, the challenges in designing the optimal multicast scheduling come from the heterogeneous structure of the network, and the difficulty in balancing the delay cost and the power cost.
In the sequel, we formalize the multicast scheduling problem and try to tackle these challenges.

\section{Problem Formulation and Optimality Equation}
\subsection{Problem Formulation}
 Given an observed request queue state $\mathbf{Q}$, the multicast scheduling action $\mathbf{u}$ is determined according to a stationary policy defined below.
\begin{definition}[Stationary Policy]\label{definition:stationary_policy}
A feasible stationary multicast scheduling policy $\mu$ is a mapping from the request queue state $\mathbf{Q}\in\mathbf{\mathcal{Q}}$ to the feasible multicast scheduling action $\mathbf{u}\in\mathbf{\mathcal{U}}$, where $\mu(\mathbf{Q})=\mathbf{u}$.
\end{definition}

By the queue dynamics in \eqref{eqn:queue1} and \eqref{eqn:queue2}, the induced random process $\{\mathbf{Q}(t)\}$ under policy $\mu$ is a controlled Markov chain with the following transition probability:
\begin{align}\label{eqn:tranb}
  &\Pr[\mathbf{Q}'|\mathbf{Q},\mathbf{u}]
\triangleq\Pr[\mathbf{Q}(t+1)=\mathbf{Q}'|\mathbf{Q}(t)=\mathbf{Q},\mathbf{u}(t)=\mathbf{u}]\nonumber\\
&=\mathbb{E}\left[\Pr\left[\mathbf{Q}(t+1)=\mathbf{Q}'|\mathbf{Q}(t)=\mathbf{Q},\mathbf{u}(t)=\mathbf{u},\mathbf{A}(t)=\mathbf{A}\right]\right],
\end{align}
where the expectation is taken over the distribution of the request arrival $\mathbf{A}$ and
\begin{align*}
\Pr\left[\mathbf{Q}(t+1)=\mathbf{Q}'|\mathbf{Q}(t)=\mathbf{Q},\mathbf{u}(t)=\mathbf{u},\mathbf{A}(t)=\mathbf{A}\right]\nonumber\\
=\left\{		
   \begin{array}{ll}
       1, & \hbox{if $\mathbf{Q}'$ satisfies \eqref{eqn:queue1} and \eqref{eqn:queue2}} \\
       0, & \hbox{otherwise}
   \end{array}
 \right..
\end{align*}

We restrict our attention to stationary unichain policies.\footnote{A unichain policy is a policy, under which the induced Markov chain has a single recurrent class (and possibly some transient states)\cite{bertsekas}.}
For a given stationary unichain policy $\mu$, the average network delay cost is defined as
\begin{equation}
  \bar{d}(\mu)\triangleq\limsup_{T\to\infty}\frac{1}{T}\sum_{t=1}^T \mathbb{E}\left[d\left(\mathbf{Q}(t)\right)\right],\label{eqn:delay_cost}
\end{equation}
where $d\left(\mathbf{Q}\right)\triangleq\sum_{n\in\mathcal{N}}\sum_{m\in\mathcal{M}_n}Q_{n,m}$ and the expectation is taken with respect to the measure induced by the random request arrivals and the policy $\mu$. According to Little's law, $\bar{d}(\mu)$ reflects the average waiting time in the network under policy $\mu$.
For a given stationary unichain policy $\mu$, the average network power cost is given by
\begin{equation}
  \bar{p}(\mu)\triangleq\limsup_{T\to\infty}\frac{1}{T}\sum_{t=1}^T \mathbb{E}\left[p(\mathbf{u}(t))\right].\label{eqn:power_cost}
\end{equation}
In this paper, we would like to jointly minimize the average network delay cost and the average network power cost. We adopt the weighted-sum method, which is a commonly used  method for multiobjective optimization problems \cite{deb2014multi}. Specifically, we define the average network cost as the weighted sum of the average network delay cost and the average network power cost, i.e.,
\begin{align}
  \bar{g}(\mu)\triangleq\bar{d}(\mu)+w\bar{p}(\mu)
  =\limsup_{T\to\infty}\frac{1}{T}\sum_{t=1}^T \mathbb{E}\left[g(\mathbf{Q}(t),\mathbf{u}(t))\right],\label{eqn:system_cost}
\end{align}
where $g(\mathbf{Q},\mathbf{u})\triangleq d\left(\mathbf{Q}\right)+wp(\mathbf{u})$ is the per-stage network cost and $w$ is the weight indicating  the relative importance of the average power cost over the average delay cost.  Note that, $w$ can also be treated as the penalty factor, mimicking the soft average delay constraint. In other words,  $w$ can be thought of a Lagrange multiplier on the average delay cost constraint\cite{berry2002communication}.

We wish to find an optimal multicast scheduling policy to minimize the average network cost $\bar{g}(\mu)$ in \eqref{eqn:system_cost}.
\begin{problem}[Network Cost Minimization]\label{problem:originalproblem}
\begin{equation}
  \bar{g}^*\triangleq\min_{\mu}\limsup_{T\to\infty}\frac{1}{T}\sum_{t=1}^T \mathbb{E}\left[g(\mathbf{Q}(t),\mathbf{u}(t))\right],\label{eqn:problem1}
\end{equation}
where $\mu$ is a stationary unichain policy  in Definition~\ref{definition:stationary_policy} and $\bar{g}^*$ denotes the minimum average network cost achieved by the optimal policy $\mu^*$.\footnote{Restricting the search for the optimal policy to unichain policies is to guarantee the existence of stationary optimal policies. This is widely used in the literature, e.g.,\cite{6665021,djonin2007mimo,levorato2009optimal}.}
\end{problem}

 Notice that, if $w=0$, Problem~\ref{problem:originalproblem} reduces to the delay minimization problem, which can still be covered by our framework.  However, if the power minimization is the sole goal, this optimization problem is no
longer meaningful, as all BSs would keep idle and no requests would be served.
In Problem~\ref{problem:originalproblem}, we assume the existence of a stationary unichain policy achieving the minimum in \eqref{eqn:problem1}. Later in Lemma~\ref{lemma:bellman}, we shall prove the existence of such a policy.
Problem~\ref{problem:originalproblem} is an infinite horizon average cost MDP, which is  challenging due to the curse of dimensionality.

\subsection{Optimality Equation}
The optimal multicast scheduling policy $\mu^*$ can be obtained by solving the following Bellman equation.
\begin{lemma}[Bellman Equation]\label{lemma:bellman}
 There exist a scalar $\theta$ and a value function $V(\cdot)$ satisfying:
\begin{equation}
  \theta+V(\mathbf{Q})=\min_{\mathbf{u}\in\mathbf{\mathcal{U}}}\left\{g(\mathbf{Q},\mathbf{u})+\mathbb{E}\left[V(\mathbf{Q}')\right]\right\},~\forall \mathbf{Q}\in\mathbf{\mathcal{Q}},\label{eqn:bellman}
\end{equation}
where the expectation is taken over the distribution of the request arrival $\mathbf{A}$, and $\mathbf{Q}'\triangleq(Q'_{n,m})_{n\in\mathcal{N},m\in\mathcal{M}_n}$ with $Q'_{0,m}\triangleq\min\{\mathbf{1}(u_0\neq m)Q_{0,m}+\tilde{A}_{0,m},N_{0,m}\}$ for all $m\in\mathcal{M}_0$ and $Q'_{n,m}\triangleq\min\{\mathbf{1}(u_0\neq m~\&~u_n\neq m)Q_{n,m}+A_{n,m},N_{n,m}\}$ for all $n\in\mathcal{N}^+$ and $m\in\mathcal{M}_n$.
Furthermore, $\theta=\bar{g}^*$ is the optimal value to Problem~\ref{problem:originalproblem} for all initial state $\mathbf{Q}(1)\in\mathbf{\mathcal{Q}}$ and the optimal policy achieving the optimal value $\bar{g}^*$ is given by
\begin{eqnarray}
  \mu^*(\mathbf{Q})=\arg\min_{\mathbf{u}\in\mathbf{\mathcal{U}}}\left\{g(\mathbf{Q},\mathbf{u})+\mathbb{E}\left[V(\mathbf{Q}')\right]\right\},~\forall \mathbf{Q}\in\mathbf{\mathcal{Q}}.\label{eqn:mu}
\end{eqnarray}
\end{lemma}

\begin{proof}
   Please see Appendix A.
\end{proof}

From Lemma~\ref{lemma:bellman}, we observe that $\mu^*$ given by \eqref{eqn:mu} depends on $\mathbf{Q}$ through the value function $V(\cdot)$.
Obtaining $V(\cdot)$ involves solving the Bellman equation in \eqref{eqn:bellman} for all $\mathbf{Q}\in\mathbf{\mathcal{Q}}$, which does not admit a closed-form solution in general \cite{bertsekas}.Standard numerical solutions such as value iteration and policy iteration are usually computationally impractical to implement, and do not typically yield many design insights\cite{bertsekas}. Thus, it is highly desirable to study the structural properties of the optimal policy $\mu^*$.

\section{Optimality Properties}\label{sec:optimality}
Problem 1 can be viewed as a problem of scheduling a single broadcast server (the MBS) or multiple multicast servers (the SBSs) to parallel (request) queues with general arrivals. The structural analysis is more challenging than the existing structural analysis for the scheduling of a single broadcast server.
First, by RVIA and the special structure of the request queue dynamics, we can prove the following property of the value function.
\begin{lemma}[Monotonicity of $V(\mathbf{Q})$]\label{lemma:propertyV}
For any $\mathbf{Q}^1$,$\mathbf{Q}^2\in\mathbf{\mathcal{Q}}$ such that $\mathbf{Q}^2\succeq\mathbf{Q}^1$, we have $V(\mathbf{Q}^2)\geq V(\mathbf{Q}^1)$.\footnote{The notation $\succeq$ indicates component-wise $\geq$.}
\end{lemma}
\begin{proof}
  Please see Appendix B.
\end{proof}

Next, we introduce the state-action cost function:
\begin{equation}
J(\mathbf{Q},\mathbf{u})\triangleq g(\mathbf{Q},\mathbf{u})+\mathbb{E}\left[V(\mathbf{Q}')\right].
\end{equation}
Note that $J(\mathbf{Q},\mathbf{u})$ is related to the R.H.S. of the Bellman equation in \eqref{eqn:bellman}.
Then, based on $J(\mathbf{Q},\mathbf{u})$, we introduce:
\begin{equation}
  \Delta_{\mathbf{u},\mathbf{v}}(\mathbf{Q})\triangleq J(\mathbf{Q},\mathbf{u})-J(\mathbf{Q},\mathbf{v}).\label{eqn:delta_func}
\end{equation}
 Note that $\Delta_{\mathbf{u},\mathbf{v}}(\mathbf{Q})=-\Delta_{\mathbf{v},\mathbf{u}}(\mathbf{Q})$.
Action $\mathbf{u}$ is said to dominate $\mathbf{v}$ at state $\mathbf{Q}$ if $\Delta_{\mathbf{u},\mathbf{v}}(\mathbf{Q})\leq 0$.
In particular, by Lemma~\ref{lemma:bellman}, if $\mathbf{u}$ dominates all $\mathbf{v}\in\mathbf{\mathcal{U}}$ at state $\mathbf{Q}$, then $\mu^*(\mathbf{Q})=\mathbf{u}$.
Based on Lemma~\ref{lemma:propertyV},  we have the following property of the function defined in \eqref{eqn:delta_func}.
\begin{lemma}[Monotonicity of $\Delta_{\mathbf{u},\mathbf{v}}(\mathbf{Q})$]\label{lemma:property_delta}
For any $\mathbf{Q}\in\mathbf{\mathcal{Q}}$ and $\mathbf{u},\mathbf{v}\in\mathbf{\mathcal{U}}$, $\Delta_{\mathbf{u},\mathbf{v}}(\mathbf{Q})$ has the following properties.



\begin{enumerate}
\item If $u_0\in\mathcal{M}_0$, then $\Delta_{\mathbf{u},\mathbf{v}}(\mathbf{Q})$ is monotonically non-increasing with $Q_{0,m}$ and $Q_{n,m}$ for all $n\in\mathcal{N}_{m}$, where $m=u_0$.

\item If $u_n\in\mathcal{M}_n$ for some $n\in\mathcal{N}^+$, then $\Delta_{\mathbf{u},\mathbf{v}}(\mathbf{Q})$ is monotonically non-increasing with  $Q_{n,m}$,  where $m=u_n$.
\end{enumerate}
\end{lemma}
\begin{proof}
   Please see Appendix C.
\end{proof}

Lemma~\ref{lemma:property_delta} indicates that, if $\mathbf{u}$ dominates $\mathbf{v}$ at some state $\mathbf{Q}$, then by increasing $Q_{0,u_0}$ or $Q_{n,u_0}$ for any $n\in\mathcal{N}_{u_0}$ and $u_0\neq 0$, or by increasing $Q_{n,u_n}$ for any $n\in\mathcal{N}^+$ and $u_n\neq 0$, $\mathbf{u}$ still dominates $\mathbf{v}$.
The properties of $\Delta_{\mathbf{u},\mathbf{v}}(\mathbf{Q})$ in Lemma~\ref{lemma:property_delta} are similar to the diminishing-return property of submodular functions used in the existing structural analysis \cite{Koole}. Lemma~\ref{lemma:property_delta} stems from the special properties of multicasting and is essential to characterize the optimality properties.
By Lemma~\ref{lemma:property_delta}, we can characterize the structural properties of the optimal policy $\mu^*$. We start with several definitions.
Define
\begin{align}
\Phi_{\mathbf{u}}&(\mathbf{Q}_{-n,-m})\triangleq\{Q_{n,m}|Q_{n,m}\in\mathcal{Q}_{n,m}~\text{and}\nonumber\\&\Delta_{\mathbf{u},\mathbf{v}}(Q_{n,m},\mathbf{Q}_{-n,-m})\leq 0~\forall \mathbf{v}\in\mathbf{\mathcal{U}}~\text{and}~\mathbf{v}\neq \mathbf{u}\},
\end{align}
where
$\mathbf{Q}_{-n,-m}$
\noindent$\triangleq (Q_{i,j})_{i\in\mathcal{N},j\in\mathcal{M}_i,(i,j)\neq(n,m)}$.
Based on $\Phi_{\mathbf{u}}(\cdot)$, we define:
\begin{align}
&\phi_{\mathbf{u}}^+(\mathbf{Q}_{-n,-m})\triangleq\begin{cases}\max\Phi_{\mathbf{u}}(\mathbf{Q}_{-n,-m}),  & \text{if}~\Phi_{\mathbf{u}}(\mathbf{Q}_{-n,-m})\neq\emptyset \\
            -\infty,  &\text{otherwise}
  \end{cases},\label{eqn:phi}\\
&\phi_{\mathbf{u}}^-(\mathbf{Q}_{-n,-m})
\triangleq\begin{cases}\min\Phi_{\mathbf{u}}(\mathbf{Q}_{-n,-m}),  & \text{if}~\Phi_{\mathbf{u}}(\mathbf{Q}_{-n,-m})\neq\emptyset \\
            +\infty,  &\text{otherwise}
  \end{cases}.\label{eqn:psi}
\end{align}
Let $\mathbf{0}$ denote the $1\times (N+1)$ vector with all entries $0$.
Then, we have the following theorem.
\begin{theorem}[Structural properties of $\mu^*$] \label{theorem:optimal}
For any $\mathbf{Q}\in\mathbf{\mathcal{Q}}$, the optimal policy $\mu^*$ has the following structural properties.


\begin{enumerate}
  \item $\mu^*(\mathbf{Q})=\mathbf{0}$ for all $\mathbf{Q}\in\mathbf{\mathcal{Q}}_0\triangleq\{\mathbf{Q}|Q_{n,m}\leq \phi_{\mathbf{0}}^+(\mathbf{Q}_{-n,-m}), \forall m\in\mathcal{M}_n~\text{and}~n\in\mathcal{N}\}$.
\item If $\exists n\in\mathcal{N}$, such that $u_n^*\in\mathcal{M}_n$, then $\mu^*(\mathbf{Q})=\mathbf{u}^*$ for all $\mathbf{Q}\in\mathbf{\mathcal{Q}}$ such that
 \begin{equation}
Q_{n,m}\geq \phi_{\mathbf{u}^*}^-(\mathbf{Q}_{-n,-m}),
 \end{equation}
where $m=u_n^*$.
 Moreover, $\phi_{\mathbf{u}^*}^-(\mathbf{Q}_{-0,-m})$ is monotonically  non-increasing with $Q_{n,m}$ for all $n\in\mathcal{N}_{m}$.
\end{enumerate}
\end{theorem}

\begin{figure*}[!t]
\begin{minipage}[t]{0.245\linewidth}
\centering
\includegraphics[scale=0.325]{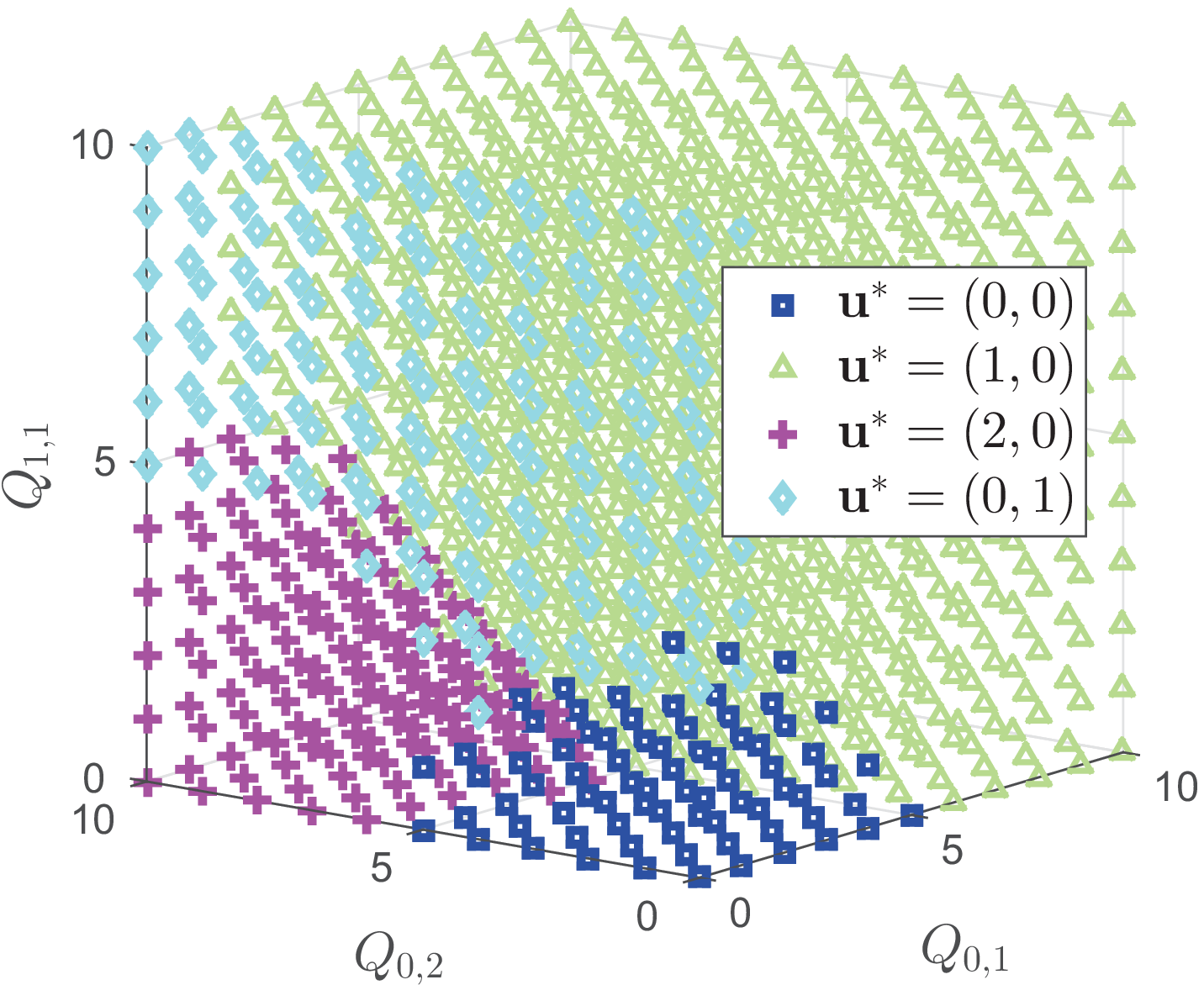}
\subcaption{Whole space.}\label{fig:3D}
\end{minipage}%
\begin{minipage}[t]{.245\linewidth}
\centering
\includegraphics[scale=0.33]{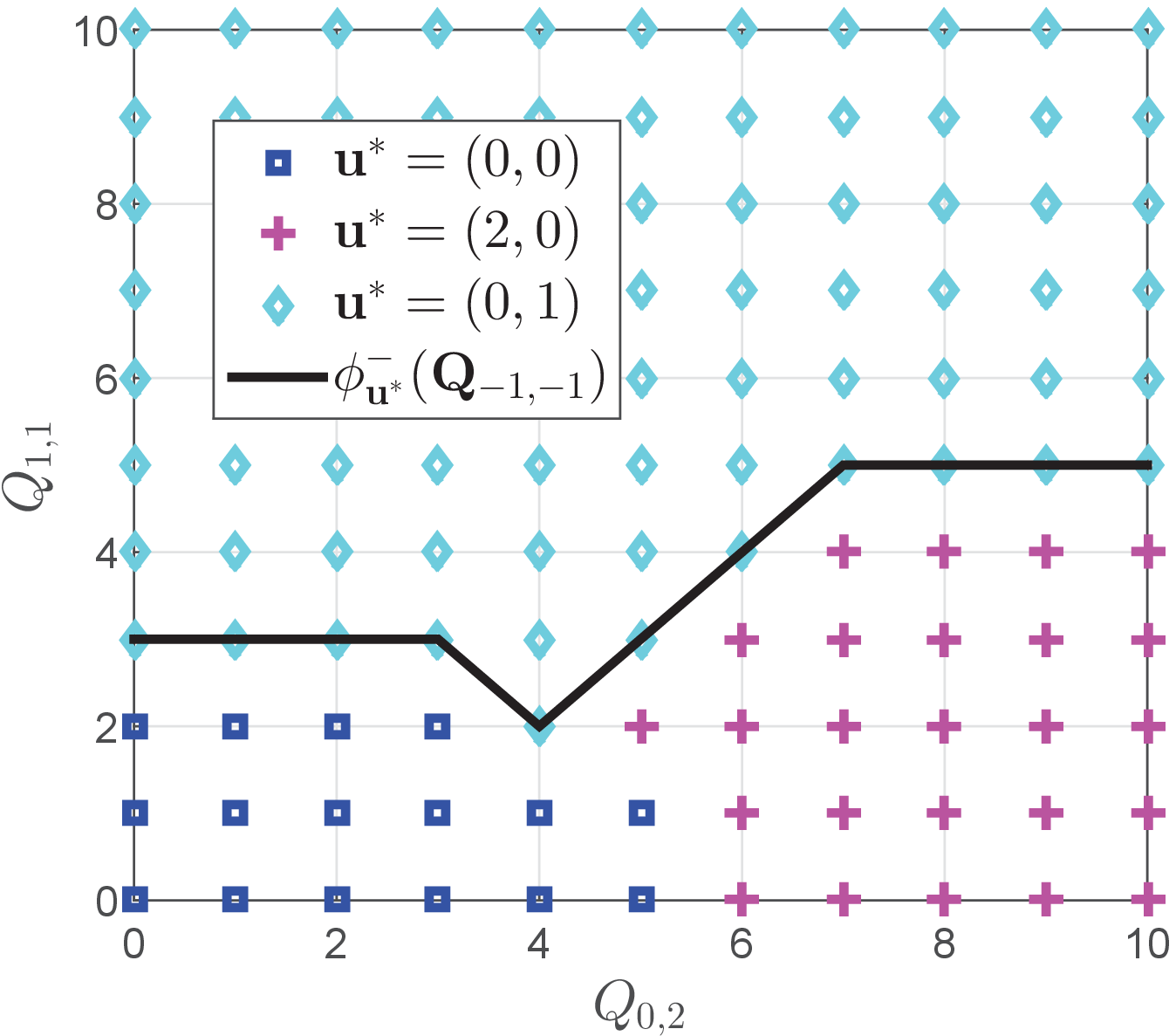}
\subcaption{$Q_{0,1}=0$.}\label{fig:2D_fixed_Q01}
\end{minipage}
\begin{minipage}[t]{.245\linewidth}
\centering
\includegraphics[scale=0.33]{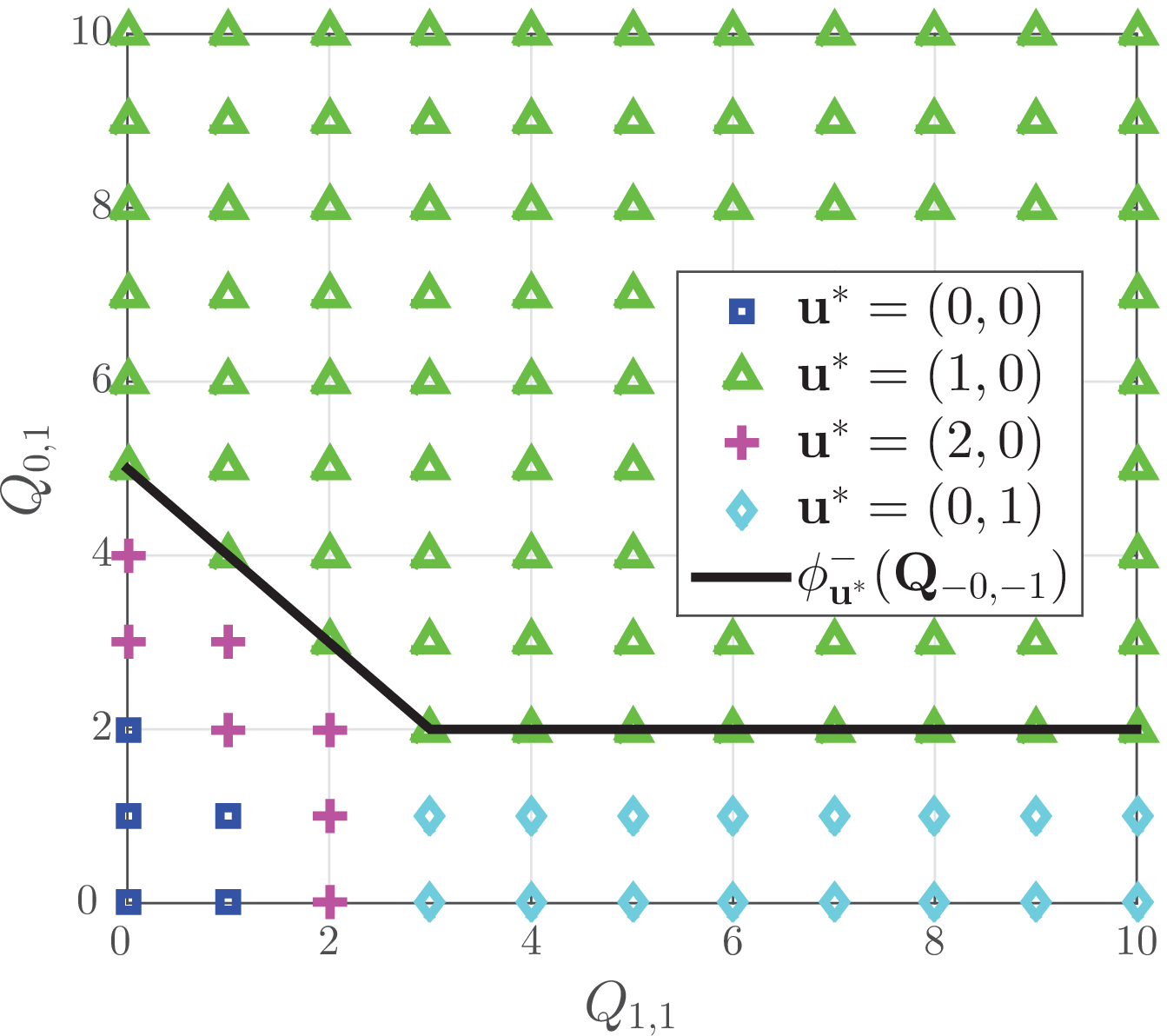}
\subcaption{$Q_{0,2}=5$.}\label{fig:2D_fixed_Q02}
\end{minipage}
\begin{minipage}[t]{.245\linewidth}
\centering
\includegraphics[scale=0.33]{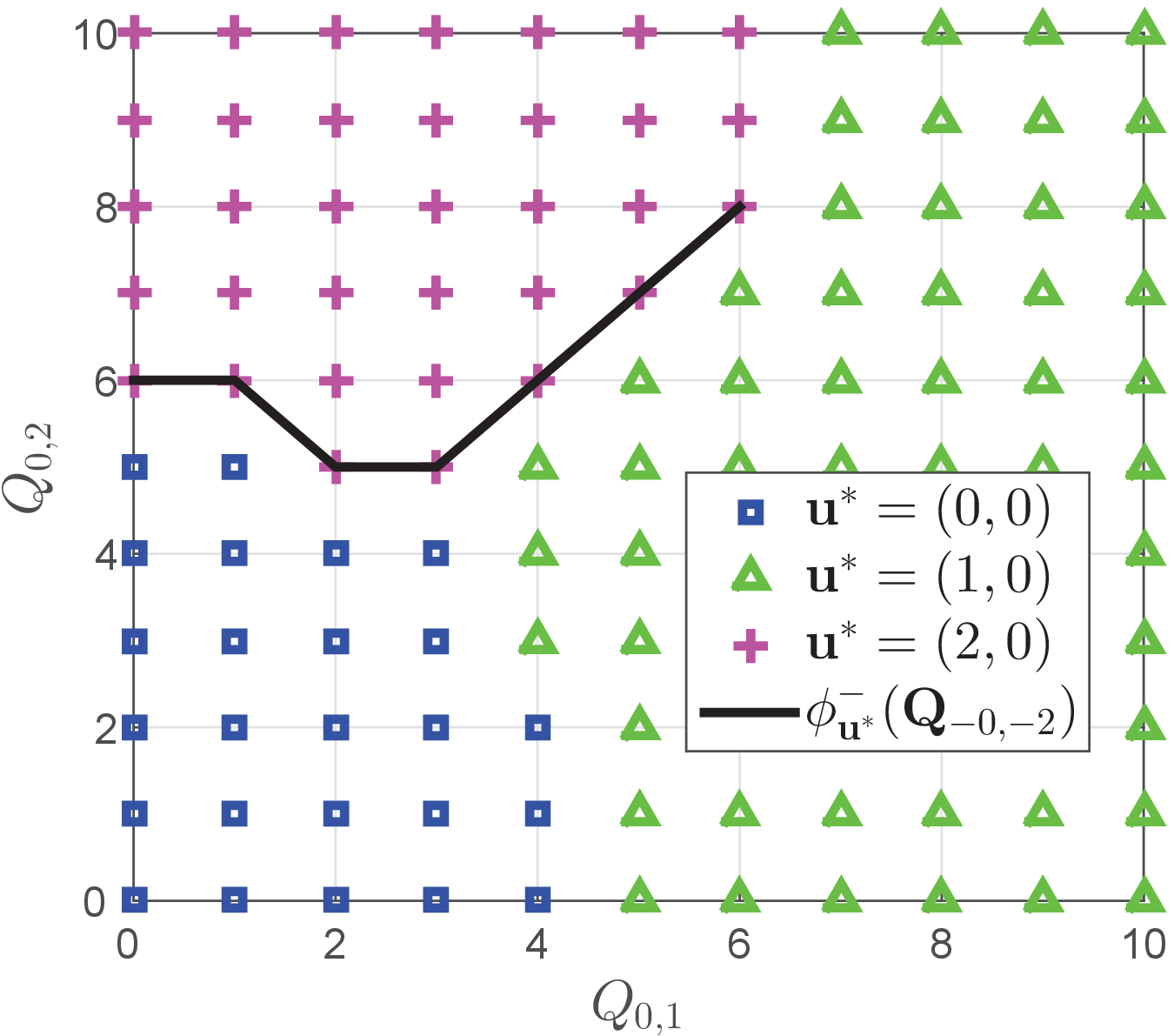}
\subcaption{$Q_{1,1}=1$.}\label{fig:2D_fixed_Q11}
\end{minipage}
\caption{Structure of the optimal multicast scheduling. $N=1$, $\mathcal{K}=\{1,2,3,4\}$, $\mathcal{M}=\{1,2\}$, $\mathcal{K}_1=\{1,2\}$, and $\mathcal{M}_1=\{1\}$. In each slot, each user requests one content, which is content 1 with probability 0.6 and content 2 with probability 0.4.}\label{fig:structure}
\end{figure*}

\begin{proof}
  Please see Appendix D.
\end{proof}

We illustrate the analytical results of Theorem~\ref{theorem:optimal} in Fig.~\ref{fig:structure}, where the optimal policy is computed numerically using \emph{policy iteration algorithm} (PIA)\cite[Chapter 8.6]{puterman}.
We observe from Fig.~\ref{fig:3D} that, if the queue state falls in the region of blue squares (i.e., $\mathbf{\mathcal{Q}}_0$), the optimal control is $(0,0)$, i.e., both the MBS and the SBS keep idle. Hence, we refer to $\mathbf{\mathcal{Q}}_0$ as the idle region of the optimal policy.
From Fig.~\ref{fig:2D_fixed_Q01}-\ref{fig:2D_fixed_Q11}, we observe that given $\mathbf{Q}_{-n,-m}$, the scheduling for content $m\in\mathcal{M}_n$ by BS $n\in\mathcal{N}$ is of threshold type (Property 2 of Theorem~\ref{theorem:optimal}).
This indicates that, it is not efficient to schedule content $m$ by BS $n$ when $Q_{n,m}$ is small (i.e., the delay cost is small), as a higher power cost per request is consumed. This shows the tradeoff between the delay cost and the power cost.
Fig.~\ref{fig:2D_fixed_Q02} illustrates the monotonically non-increasing property of $\phi_{\mathbf{u}^*}^-(\mathbf{Q}_{-0,-1})$ in terms of  $Q_{1,1}$. This reveals that the MBS is more willing to multicast content $1$ when $Q_{1,1}$ is large. The reason is that the MBS can satisfy more requests than any SBS.
These optimality properties provide design insights for multicast scheduling in practical
cache-enabled HetNets.


\section{Structure-Aware Optimal Algorithm}\label{sec:SPIA}
The results in Theorem~\ref{theorem:optimal} can be exploited to substantially reduce the computational complexity for solving the Bellman equation in \eqref{eqn:bellman} in obtaining $\mu^*$. In particular, by Property 2 in Theorem~\ref{theorem:optimal}, for all $\mathbf{Q}\in\mathbf{\mathcal{Q}}$, $\mathbf{u}\in\mathbf{\mathcal{U}}$ and $\mathbf{Q}'=(Q_{n,m}')_{n\in\mathcal{N},m\in\mathcal{M}_n}$ satisfying
\begin{equation}
\begin{cases}
Q_{n,m}'\geq Q_{n,m},  & \text{if}~u_n=m \\
Q_{n,m}'= Q_{n,m},  &\text{if}~u_n\neq m
  \end{cases}
\end{equation}
for each $n\in\mathcal{N}$ and $m\in\mathcal{M}_n$,  we have
\begin{equation}\label{eqn:indication}
  \mu^*(\mathbf{Q})=\mathbf{u}~\Rightarrow~\mu^*(\mathbf{Q}')=\mathbf{u}.
\end{equation}
Therefore, by incorporating the property in \eqref{eqn:indication} into the standard PIA, we develop a structure-aware algorithm in Algorithm~\ref{alg:SPIA}, which is referred to as the structured policy iteration algorithm (SPIA).
According to Theorem 8.6.6 and Chapter 8.11.2 in \cite{puterman}, we know that SPIA converges to the optimal policy $\mu^*$ in \eqref{eqn:mu} within a finite number of iterations, and hence is an optimal algorithm.

\begin{algorithm}[!t]
\caption{Structured Policy Iteration Algorithm}
\label{alg:SPIA}
\begin{algorithmic}[1]
\State Set $\mu^*_0(\mathbf{Q})=\mathbf{0}$ for all $\mathbf{Q}\in\mathbf{\mathcal{Q}}$, select reference state $\mathbf{Q}^\dag$, and set $l=0$.
\State (Policy Evaluation) Given policy $\mu_l^*$, compute the average cost $\theta_l$ and value function $V_l(\mathbf{Q})$ from the linear system of equations\footnotemark
\begin{equation}\label{eqn:spia_eva}
\begin{cases}
     \theta_l+V_l(\mathbf{Q})=g(\mathbf{Q},\mu_l^*(\mathbf{Q}))+\mathbb{E}\left[V_l(\mathbf{Q}')\right],~\forall \mathbf{Q}\in\mathbf{\mathcal{Q}}\\
  V_l(\mathbf{Q}^\dag)=0
\end{cases},
\end{equation}
\Statex where $\mathbf{Q}'$ is defined in Lemma~\ref{lemma:bellman}.\label{code:spia_eva}
\State (Structured Policy Improvement) Obtain a new policy $\mu_{l+1}^*$, where for each $\mathbf{Q}\in\mathbf{\mathcal{Q}}$, $\mu_{l+1}^*(\mathbf{Q})$ is such that:\label{code:spia_imp}
\Statex\textbf{if}~$\exists n\in\mathcal{N}$ and $\mathbf{Q}'\in\mathbf{\mathcal{Q}}$ such that $\mu_l^*(\mathbf{Q}')=\mathbf{u}$, $u_n\in\mathcal{M}_n$, $Q_{n,m}'\leq Q_{n,m}$, and $Q_{i,j}'=Q_{i,j}$ for all $(i,j)\neq (n,m)$, where $m=u_n$~\textbf{then}
        \begin{equation*}
         \mu_{l+1}^*(\mathbf{Q})=\mathbf{u}.
        \end{equation*}
\Statex\textbf{else}
        \begin{equation*}
            \mu_{l+1}^*(\mathbf{Q})=\arg\min_{\mathbf{u}\in\mathbf{\mathcal{U}}}\left\{g(\mathbf{Q},\mathbf{u})+\mathbb{E}\left[V_l(\mathbf{Q}')\right]\right\}.
        \end{equation*}
\Statex\textbf{endif}
\State Go to Step \ref{code:spia_eva} until $\mu_{l+1}^*=\mu_{l}^*$.
\end{algorithmic}
\end{algorithm}
\footnotetext{The solution to \eqref{eqn:spia_eva} can be obtained directly using Gaussian elimination or iteratively using the relative value iteration method\cite{bertsekas}.}
Note that, in Step~\ref{code:spia_imp} (structured policy improvement) of Algorithm~\ref{alg:SPIA}, we do not need to perform the minimization over $\mathbf{\mathcal{U}}$ when the condition is satisfied (which is the case for a large amount of queue states in $\mathbf{\mathcal{Q}}$). This can be seen in Fig.~\ref{fig:structure}.
While, in the standard policy improvement step of PIA, the new policy $\mu_{l+1}^*$is obtained by:
    \begin{eqnarray}
      \mu_{l+1}^*(\mathbf{Q})=\arg\min_{\mathbf{u}\in\mathbf{\mathcal{U}}}\left\{g(\mathbf{Q},\mathbf{u})+\mathbb{E}\left[V_l(\mathbf{Q}')\right]\right\},~\forall\mathbf{Q}\in\mathbf{\mathcal{Q}}.\label{eqn:pia_imp}
    \end{eqnarray}
By \eqref{eqn:pia_imp}, obtaining $\mu_{l+1}^*$ requires a brute-force minimization over  $\mathbf{\mathcal{U}}$ for each $\mathbf{Q}\in\mathbf{\mathcal{Q}}$, which can be very computationally expensive when the numbers of the contents $M$ and the SBSs $N$ are large.
By comparing the structured policy improvement step of SPIA with the standard policy improvement step of PIA, we can see that SPIA can achieve considerable computational saving.
Note that, the complexity of SPIA depends on the specific algorithm used for solving the linear system equations in \eqref{eqn:spia_eva}. For instance, using Gaussian elimination for solving \eqref{eqn:spia_eva}, the complexity of each iteration in SPIA is $O(|\mathbf{\mathcal{Q}}|^3 + |\mathbf{\mathcal{U}}||\mathbf{\mathcal{Q}}|^2)$ \cite{puterman}.

Although the proposed structure-aware optimal algorithm, i.e., SPIA, can alleviate the computational burden of the standard PIA, it still suffers from the curse of dimensionality\cite{bertsekas}, due to the exponential growth of the cardinality of the system state space ($|\mathbf{\mathbf{Q}}|=\prod_{n\in\mathcal{N}}\prod_{m\in\mathcal{M}_n}|\mathcal{Q}_{n,m}|$). When the number of the SBSs and the cache sizes are large, the resulting huge state space may render SPIA computationally impractical.

\section{Low Complexity Suboptimal Solution}\label{sec:suboptimal}
To further reduce the complexity of the proposed structure-aware optimal algorithm (i.e., SPIA) and relieve the curse of dimensionality, we would like to develop low-complexity suboptimal solutions. Note that the structural properties of the optimal policy may be one key reason that leads to good performance. Thus, in this section, we focus on design of suboptimal solutions which can maintain the structures of the optimal policy. Specifically,
 based on a randomized base policy, we first propose a low-complexity suboptimal deterministic policy using approximate dynamic programming\cite{bertsekas}. We show that the deterministic policy improves the randomized base policy and possesses similar properties to the optimal policy. Then, we design a low-complexity structured algorithm to compute the proposed policy, by exploiting these structural properties.
\subsection{Low Complexity Suboptimal Policy}\label{sec:hatmu}
The structural properties of the optimal policy come from the monotonicity property of the value function. Therefore, to maintain these structures in designing a suboptimal solution, we consider a value function decomposition method that can preserve the structural properties of the value function. Based on this decomposition, we propose a low-complexity suboptimal deterministic policy, which will be shown to possess similar structural properties to the optimal policy.
We first introduce a randomized base policy.
\begin{definition}[Randomized Base Policy]\label{def:randomized}
A randomized base policy for the multicast scheduling control $\hat{\mu}$ is given by a distribution on the feasible multicast scheduling action space $\mathbf{\mathcal{U}}$.
\end{definition}

We restrict our attention to randomized unichain base policies. Let $\hat{\theta}$ and $\hat{V}(\mathbf{Q})$ denote the average cost and the value function under a randomized unichain base policy $\hat{\mu}$, respectively.
By \cite[Proposition 4.2.2]{bertsekas}, there exists $(\hat{\theta},\{\hat{V}(\mathbf{Q})\})$ satisfying:
\begin{align}
  \hat{\theta}+\hat{V}(\mathbf{Q})=\mathbb{E}^{\hat{\mu}}\left[g(\mathbf{Q},\mathbf{u})\right]+\sum_{\mathbf{Q}'\in\mathbf{\mathcal{Q}}}\mathbb{E}^{\hat{\mu}}&\left[\Pr[\mathbf{Q}'|\mathbf{Q},\mathbf{u}]\right]\hat{V}(\mathbf{Q}'),\nonumber\\
  &~~~~~~\forall \mathbf{Q}\in\mathbf{\mathcal{Q}},\label{eqn:randombellman}
\end{align}
Next, we show that $\hat{V}(\mathbf{Q})$ has a additive separable structure in the following lemma.
\begin{lemma}[Additive Separable Structure of $\hat{V}(\mathbf{Q})$]
  Given any randomized unichain base policy $\hat{\mu}$, the value function $\{\hat{V}(\mathbf{Q})\}$ in \eqref{eqn:randombellman} can be expressed as $\hat{V}(\mathbf{Q})=\sum_{n\in\mathcal{N}}\sum_{m\in\mathcal{M}_n}\hat{V}_{n,m}(Q_{n,m})$, where $\{\hat{V}_{n,m}(Q_{n,m})\}$ satisfies:
  \begin{align}
    \hat{\theta}_{n,m}&+\hat{V}_{n,m}(Q_{n,m})=\mathbb{E}^{\hat{\mu}}\left[g_{n,m}(Q_{n,m},\mathbf{u})\right]\nonumber\\
    &+\sum_{Q_{n,m}'\in\mathcal{Q}_{n,m}}\mathbb{E}^{\hat{\mu}}\left[\Pr[Q_{n,m}'|Q_{n,m},\mathbf{u}]\right]\hat{V}_{n,m}(Q_{n,m}'),\nonumber\\
    &\hspace{40mm}\forall Q_{n,m}\in\mathcal{Q}_{n,m},
\label{eqn:perrandombellman}
  \end{align}
for all $n\in\mathcal{N}$ and $m\in\mathcal{M}_n$.  Here, $\hat{\theta}_{n,m}$ and $\hat{V}_{n,m}(Q_{n,m})$ denote the per-BS-content average cost and value function under $\hat{\mu}$, respectively,   $g_{n,m}(Q_{n,m},\mathbf{u})\triangleq Q_{n,m}+w\mathbf{1}(u_n=m)p(n,m)$, and $\Pr[Q_{n,m}'|Q_{n,m},\mathbf{u}]\triangleq\Pr[Q_{n,m}(t+1)=Q_{n,m}'|Q_{n,m}(t)=Q_{n,m},\mathbf{u}(t)=\mathbf{u}]$.
\label{lemma:separable}
\end{lemma}
\begin{proof}
  Please see Appendix E.
\end{proof}

To alleviate the curse of dimensionality, we approximate the value function in \eqref{eqn:bellman} with $\hat{V}(\mathbf{Q})$:
\begin{equation}
  V(\mathbf{Q})\thickapprox\hat{V}(\mathbf{Q})=\sum_{n\in\mathcal{N}}\sum_{m\in\mathcal{M}_n}\hat{V}_{n,m}(Q_{n,m}),\label{eqn:approx}
\end{equation}
where $\{\hat{V}_{n,m}(Q_{n,m})\}$ is given by the per-BS-content fixed point equation in \eqref{eqn:perrandombellman}.
Then, we develop the following low-complexity deterministic policy $\hat{\mu}^*$:
\begin{align}
  &\hat{\mu}^*(\mathbf{Q})=\arg\min_{\mathbf{u}\in\mathbf{\mathcal{U}}}\left\{g(\mathbf{Q},\mathbf{u})+\sum_{n\in\mathcal{N}}\sum_{m\in\mathcal{M}_n}\mathbb{E}\left[\hat{V}_{n,m}(Q_{n,m})\right]\right\},\nonumber\\&\hspace{65mm}\forall \mathbf{Q}\in\mathbf{\mathcal{Q}}.\label{eqn:hatmu}
\end{align}
In the following proposition, we show that the deterministic policy $\hat{\mu}^*$ generated by \eqref{eqn:hatmu} always improves the corresponding randomized unichain base policy $\hat{\mu}$.
\begin{proposition}[Performance Improvement]
  $\hat{\theta}^*\geq\hat{\theta}$, where $\hat{\theta}$ is the average network cost under a randomized unichain base policy $\hat{\mu}$ and $\hat{\theta}^*$ is the average network cost under the proposed solution.
\end{proposition}
\begin{proof}
  This result follows directly from \cite{harvest}.
\end{proof}

Note that, to obtain $\hat{\mu}^*$ in \eqref{eqn:hatmu} via solving \eqref{eqn:perrandombellman} for all $n\in\mathcal{N}$ and $m\in\mathcal{M}_n$,  we only need to compute $\{\hat{V}_{n,m}(Q_{n,m})\}$ (a total of $O(\sum_{n\in\mathcal{N}}\sum_{m\in\mathcal{M}_n}|\mathcal{Q}_{n,m}|)$ values). The computational complexity is much lower than computing $\{V(\mathbf{Q})\}$ (a total of $O(\prod_{n\in\mathcal{N}}\prod_{m\in\mathcal{M}_n}|\mathcal{Q}_{n,m}|)$ values) via solving \eqref{eqn:bellman} in obtaining $\mu^*$ in \eqref{eqn:mu}.

\subsection{Structural Properties of Suboptimal Policy}
In this part, we investigate the structural properties of the suboptimal policy $\hat{\mu}^*$ in \eqref{eqn:hatmu}.
Along the lines of the structural analysis for the optimal policy in Section~\ref{sec:optimality}, we first introduce the state-action cost function for $\hat{\mu}^*$:
\begin{equation}
  \hat{J}(\mathbf{Q},\mathbf{u})\triangleq g(\mathbf{Q},\mathbf{u})+\sum_{n\in\mathcal{N}}\sum_{m\in\mathcal{M}_n}\mathbb{E}\left[\hat{V}_{n,m}(Q_{n,m})\right].\label{eqn:hat_J}
\end{equation}
Note that $\hat{J}(\mathbf{Q},\mathbf{u})$ is related to the R.H.S. of \eqref{eqn:hatmu}. Then, we introduce:
\begin{align}
  &\hat{\Delta}_{\mathbf{u},\mathbf{v}}(\mathbf{Q})\triangleq \hat{J}(\mathbf{Q},\mathbf{u})-\hat{J}(\mathbf{Q},\mathbf{v}),\label{eqn:hat_delta_func}\\
&\hat{\Phi}_{\mathbf{u}}(\mathbf{Q}_{-n,-m})\triangleq\{Q_{n,m}|Q_{n,m}\in\mathcal{Q}_{n,m}~\text{and}\nonumber\\&\hspace{5mm}\hat{\Delta}_{\mathbf{u},\mathbf{v}}(Q_{n,m},\mathbf{Q}_{-n,-m})\leq 0~\forall \mathbf{v}\in\mathbf{\mathcal{U}}~\text{and}~\mathbf{v}\neq \mathbf{u}\}.\label{eqn:hat_Phi}
\end{align}
By replacing $\Phi(\cdot)$ with $\hat{\Phi}(\cdot)$ in \eqref{eqn:phi} and \eqref{eqn:psi}, we have $\hat{\phi}_{\mathbf{u}}^+(\cdot)$ and $\hat{\phi}_{\mathbf{u}}^-(\cdot)$, respectively.
In the following theorem, we show that the proposed deterministic low-complexity suboptimal policy possesses similar structural properties to the optimal policy.
This similarity would be one key reason for the good performance of the proposed suboptimal policy, as will be shown in Section~\ref{sec:simulation}.
\begin{theorem}[Structural properties of $\hat{\mu}^*$] \label{theorem:suboptimal}
For any $\mathbf{Q}\in\mathbf{\mathcal{Q}}$, the optimal policy $\hat{\mu}^*$ has the following structural properties.
\begin{enumerate}
  \item $\hat{\mu}^*(\mathbf{Q})=\mathbf{0}$ for all $\mathbf{Q}\in\hat{\mathbf{\mathcal{Q}}}_0\triangleq\{\mathbf{Q}|Q_{n,m}\leq \hat{\phi}_{\mathbf{0}}^+(\mathbf{Q}_{-n,-m}), \forall m\in\mathcal{M}_n~\text{and}~n\in\mathcal{N}\}$.
 \item If $\exists n\in\mathcal{N}$, such that $\hat{u}_n^*\in\mathcal{M}_n$, then $\hat{\mu}^*(\mathbf{Q})=\hat{\mathbf{u}}^*$ for all $\mathbf{Q}\in\mathbf{\mathcal{Q}}$ such that
 \begin{equation}
Q_{n,m}\geq \hat{\phi}_{\hat{\mathbf{u}}^*}^-(\mathbf{Q}_{-n,-m}).
 \end{equation}
where $m=\hat{u}_n^*$.
 Moreover, $\hat{\phi}_{\hat{\mathbf{u}}^*}^-(\mathbf{Q}_{-0,-m})$ is monotonically  non-increasing with $Q_{n,m}$ for all $n\in\mathcal{N}_{m}$.
\end{enumerate}
\end{theorem}
\begin{proof}
  Please see Appendix F.
\end{proof}
\subsection{Structured Suboptimal Algorithm}
By employing the relation between $\hat{\mu}$ and $\hat{\mu}^*$ as well as the structural properties of $\hat{\mu}^*$ in Theorem~\ref{theorem:suboptimal}, we develop a low-complexity algorithm, referred to as the structured suboptimal algorithm (SSA), to obtain $\hat{\mu}^*$ in \eqref{eqn:hatmu}, as summarized in Algorithm~\ref{alg:SSA}. For SSA, it requires only one iteration to obtain $\hat{\mu}^*$. The complexity depends on the specific algorithm used for solving  \eqref{eqn:perrandombellman}.  For instance, using Gaussian elimination for solving \eqref{eqn:perrandombellman},  the complexity of SSA is $O(\sum_{n\in{\mathcal{N}}}\sum_{m\in\mathcal{N}}|\mathcal{Q}_{n,m}|^3 + |\mathbf{\mathcal{U}}||\mathbf{\mathcal{Q}}|^2)$.
\begin{algorithm}[!h]
\caption{Structured Suboptimal Algorithm}
\label{alg:SSA}
\begin{algorithmic}[1]
\State Given a randomized base unichain policy $\hat{\mu}$, compute the per-BS-content value function $\{\hat{V}_{n,m}(Q_{n,m})\}$ for all $n\in\mathcal{N}$ and $m\in\mathcal{M}_n$ by solving the linear system of equations in \eqref{eqn:perrandombellman}.\footnotemark \label{code:hatmu}
\State Obtain the deterministic policy $\hat{\mu}^*$, where for each $\mathbf{Q}\in\mathbf{\mathcal{Q}}$, $\hat{\mu}^*(\mathbf{Q})$ is such that:\label{code:ssa}
\Statex\textbf{if}~$\exists n\in\mathcal{N}$ and $\mathbf{Q}'\in\mathbf{\mathcal{Q}}$ such that $\hat{\mu}^*(\mathbf{Q}')=\mathbf{u}$, $u_n\in\mathcal{M}_n$, $Q_{n,m}'\leq Q_{n,m}$, and $Q_{i,j}'=Q_{i,j}$ for all $(i,j)\neq (n,m)$, where $m=u_n$~\textbf{then}
$$\hat{\mu}^*(\mathbf{Q})=\mathbf{u}.$$
\Statex\textbf{else}
\Statex \hspace{20mm}Compute  $\hat{\mu}^*(\mathbf{Q})$ using \eqref{eqn:hatmu}.
\Statex\textbf{endif}
\end{algorithmic}
\end{algorithm}

Now, we compare the computational complexity of SSA and SPIA.
SSA is similar to one iteration of SPIA. As discussed in Section~\ref{sec:hatmu}, in Step~\ref{code:hatmu} of SSA, the number of value functions required to be computed is much smaller than that in each iteration of the policy evaluation step (Step~\ref{code:spia_imp} in Algorithm~\ref{alg:SPIA}) of SPIA. In Step~\ref{code:ssa} of SSA, the number of optimizations required to solve is comparable to that in each iteration of the structured policy improvement step of SPIA.
Therefore, SSA has a significantly lower computational complexity than SPIA.
\footnotetext{The solution to \eqref{eqn:perrandombellman} can be obtained directly using Gaussian elimination or iteratively using the relative value iteration method \cite{bertsekas}.}

\section{Numerical results and discussions}\label{sec:simulation}
In this section, we evaluate the performance of the proposed optimal and suboptimal solutions through numerical examples.
In the simulations, we assume that in each slot, each user requests one content (independently), which is content $m$ with probability $P_m$. We assume that $\{P_m\}$ follows a (normalized) Zipf distribution with parameter $\alpha$ \cite{zipf}.
For simplicity, we assume that each content is of the same size and each SBS has the same cache size. These assumptions are commonly used in the literature on wireless caching in HetNets, e.g., \cite{femto,7218465,6665021,Bastug2015,7194828}. As the coverage areas of the SBSs are assumed to be disjoint, we adopt a commonly used content placement strategy, i.e., each SBS stores the most popular contents\cite{femto,Bastug2015,7194828}.
\begin{figure*}[!t]
\begin{minipage}[t]{.32\linewidth}
\centering
\includegraphics[scale=0.33]{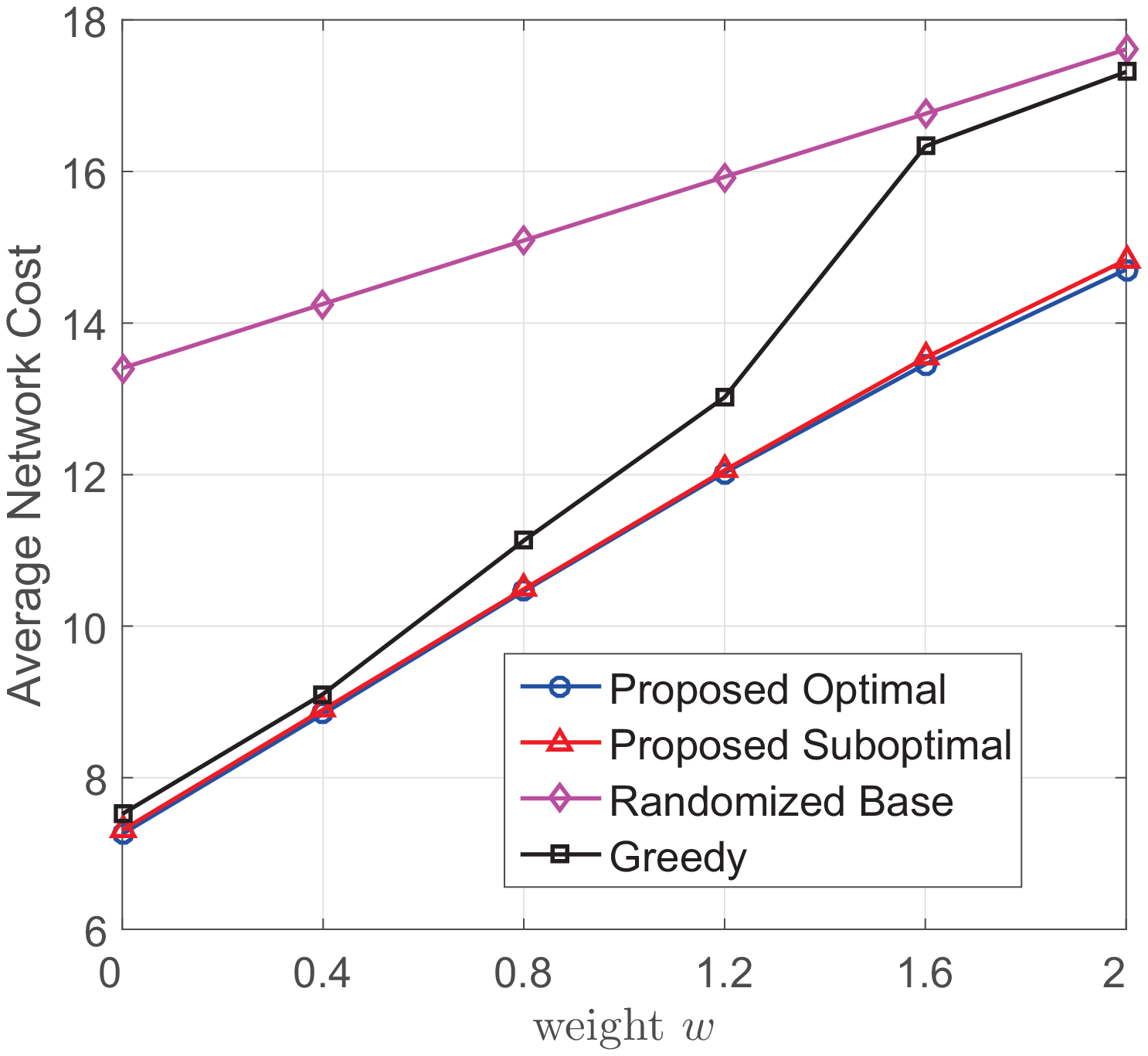}
\subcaption{Average network cost.}\label{fig:optimal_weight_cost}
\end{minipage}
\begin{minipage}[t]{.32\linewidth}
\centering
\includegraphics[scale=0.33]{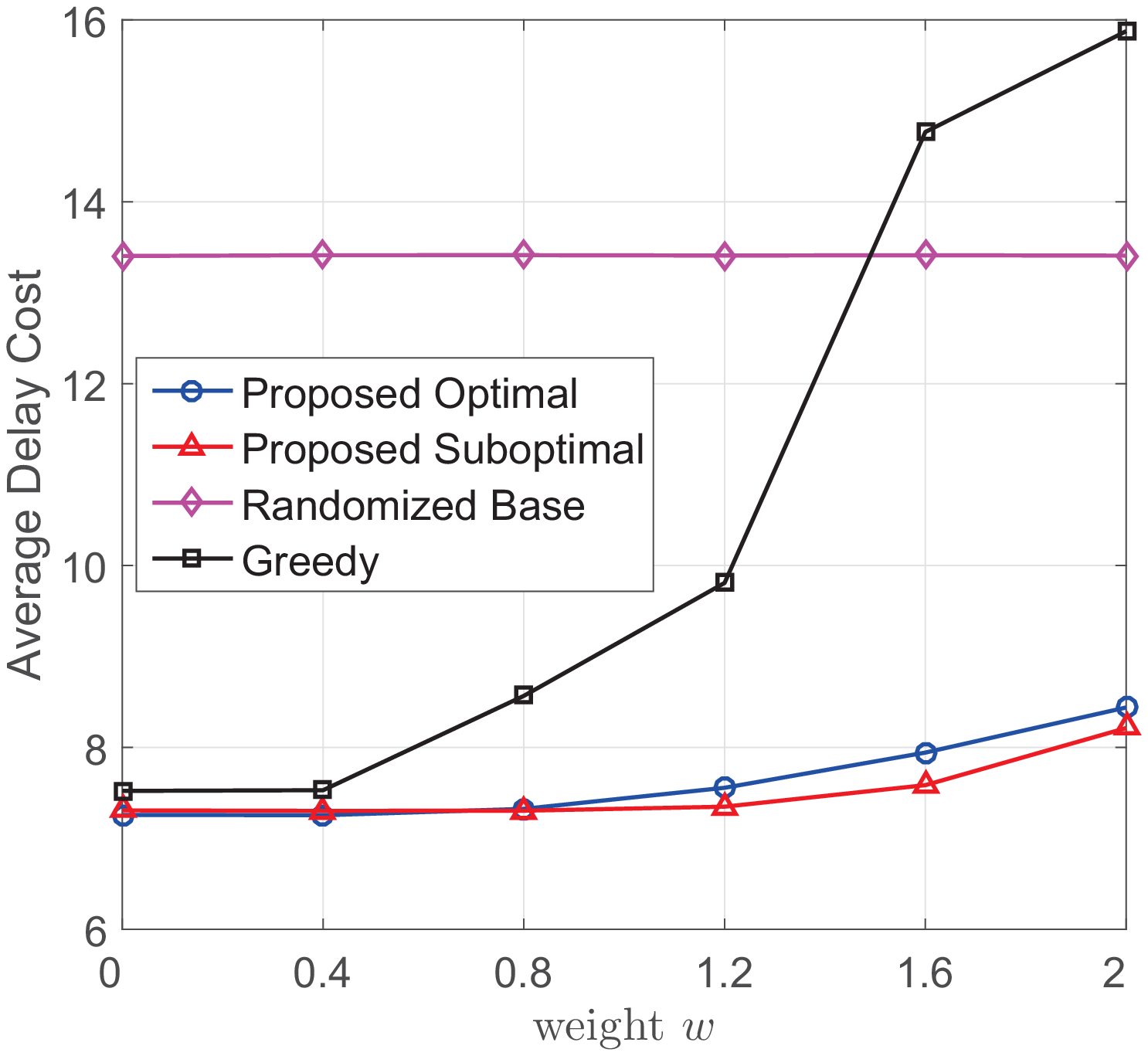}
 \subcaption{Average delay cost.}\label{fig:optimal_weight_delay}
\end{minipage}
\begin{minipage}[t]{.32\linewidth}
\centering
\includegraphics[scale=0.33]{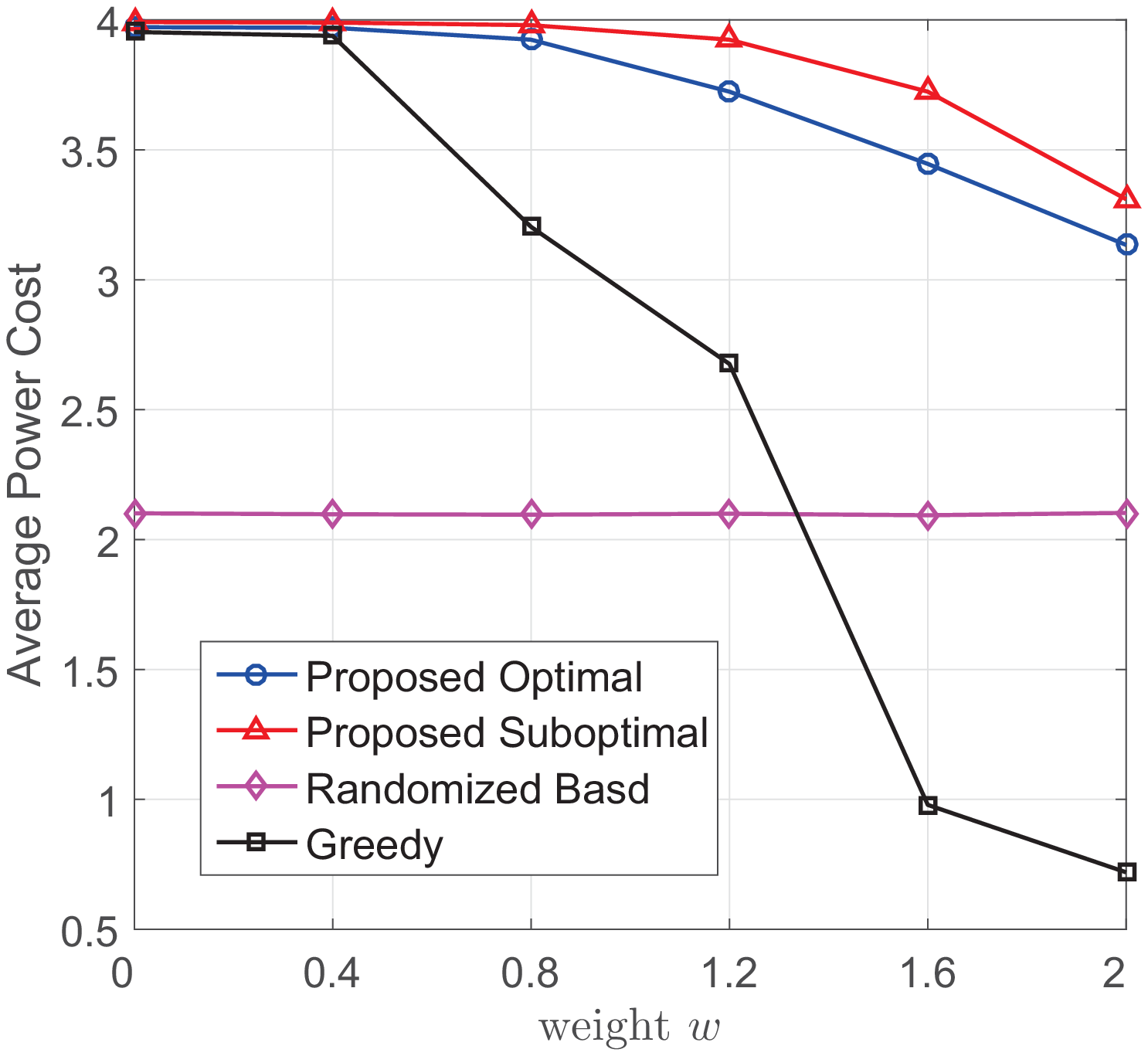}
 \subcaption{Average power cost.}\label{fig:optimal_weight_power}
\end{minipage}
\caption{Average network, delay, and power costs versus weight $w$. }\label{fig:optimal_weight}
\end{figure*}

First, we compare the average costs of the proposed optimal and suboptimal solutions with two baseline policies, i.e., a randomized base policy in Definition~\ref{def:randomized} and a greedy policy.
In particular, for the randomized base policy, in each slot, we randomly select the MBS or all the SBSs to operate, with probability  $P_{MBS}$ and  $1-P_{MBS}$, respectively.
Then, if BS $n$ is allowed to operate, it keeps idle with probability $P^n_{idle}$ or schedules one cached content $m\in\mathcal{M}_n$ for multicasting with probability $P^n_m=(1-P^n_{idle})P_m/\sum_{m\in\mathcal{M}_n}P_m$. Note that, the suboptimal policy is obtained based on this randomized base policy, as illustrated in Section~\ref{sec:suboptimal}.
In the following simulations, we set $P_{MBS}=0.5$ and $P_{idle}^n=0.3$ for all $n\in\mathcal{N}$.
To illustrate the greedy policy, we first introduce its cost function $C(\mathbf{Q},\mathbf{u})\triangleq\sum_{n\in\mathcal{N}}c_n(\mathbf{Q},u_n)$, where $\mathbf{Q}\in\mathbf{\mathcal{Q}}$, $\mathbf{u}\in\mathbf{\mathcal{U}}$, and
\begin{align*}
&c_0(\mathbf{Q},u_0)\nonumber\\
&\triangleq\begin{cases}w p(0,m) - Q_{0,m} - \sum_{n\in\mathcal{N}_m}Q_{n,m} & \text{if}~u_0=m\in\mathcal{M}_0 \\
            0  &\text{if}~u_0=0
  \end{cases},\\
&c_n(\mathbf{Q},u_n)\nonumber\\
&\triangleq\begin{cases}  w p(n,m) - Q_{n,m} & \text{if}~u_n=m\in\mathcal{M}_n \\
            0  &\text{if}~u_n=0
  \end{cases}, ~\forall n\in\mathcal{N}^+.
\end{align*}

For the greedy policy, in each slot, we choose the multicast scheduling action that minimizes the cost function, i.e.,
$\mathbf{u}(t)=\arg\min_{\mathbf{u}(t)\in\mathbf{\mathcal{U}}}C\left(\mathbf{Q}(t),\mathbf{u}(t)\right)$.
This greedy policy can also be treated as an approximate solution to Problem~\ref{problem:originalproblem} through approximating $V(\mathbf{Q})$ with $\mathbf{Q}$ in \eqref{eqn:mu}.
Note that this policy determines the scheduling action myopically, without accurately considering the impact of the action on the future costs.
This type of policy is a commonly used baseline policy in the literature (see e.g., \cite{powell2007approximate} and references therein).

Fig.~\ref{fig:optimal_weight} illustrates the average network cost, delay cost, and power cost versus the weight of the power cost (i.e., $w$) for a network with 1 MBS, 1 SBS, 4 users, and 3 contents ($\mathcal{M}=\{1,2,3\}$). We set $\alpha = 0.75$, $|\mathcal{K}_0|=|\mathcal{K}_1|=2$, $\mathcal{M}_0 = \{1,2,3\}$, $\mathcal{M}_1=\{1\}$, $p(0,m)=4$ for all $m\in\mathcal{M}_0$, and  $p(1,m)=2$ for all $m\in\mathcal{M}_1$.
It can be observed in Fig~\ref{fig:optimal_weight_cost} that the average network costs of the proposed optimal and suboptimal policies are very close to each other, and are lower than those of the randomized base policy and the greedy policy. The reason is that the proposed two policies can make foresighted decisions by better utilizing system state information and considering the immediate cost as well as the future costs.
From Fig~\ref{fig:optimal_weight_delay} and Fig~\ref{fig:optimal_weight_power}, we can see that, as $w$ increases, the average power costs of the proposed two policies  decrease, at the expense of the average delay costs. This reveals the tradeoff between the average delay cost and the average power cost.
\begin{figure*}[!t]
\begin{minipage}[t]{.32\linewidth}
\centering
\includegraphics[scale=0.36]{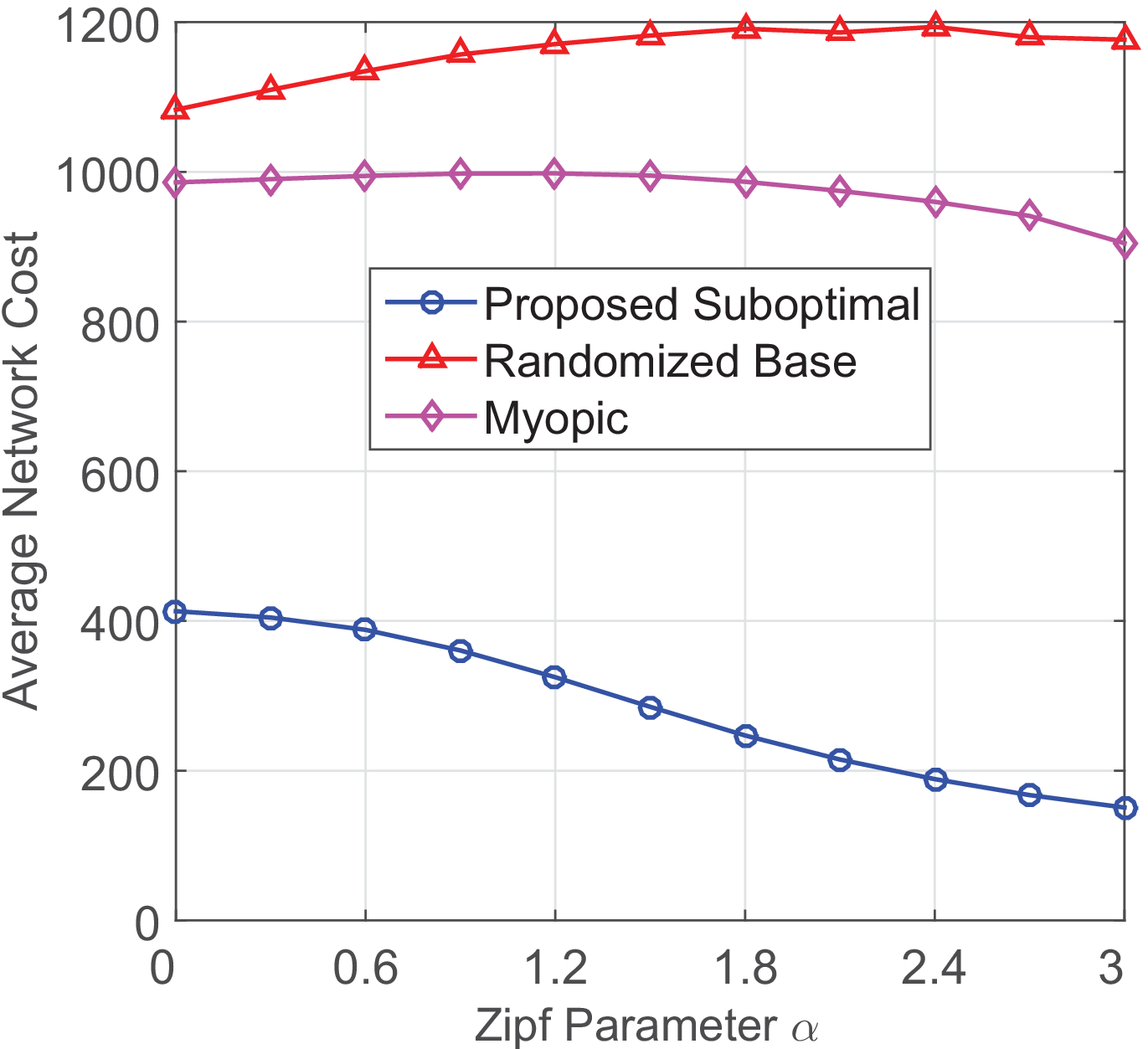}
\subcaption{Average network cost.}\label{fig:suboptimal_Zipf_cost}
\end{minipage}
\begin{minipage}[t]{.33\linewidth}
\centering
\includegraphics[scale=0.36]{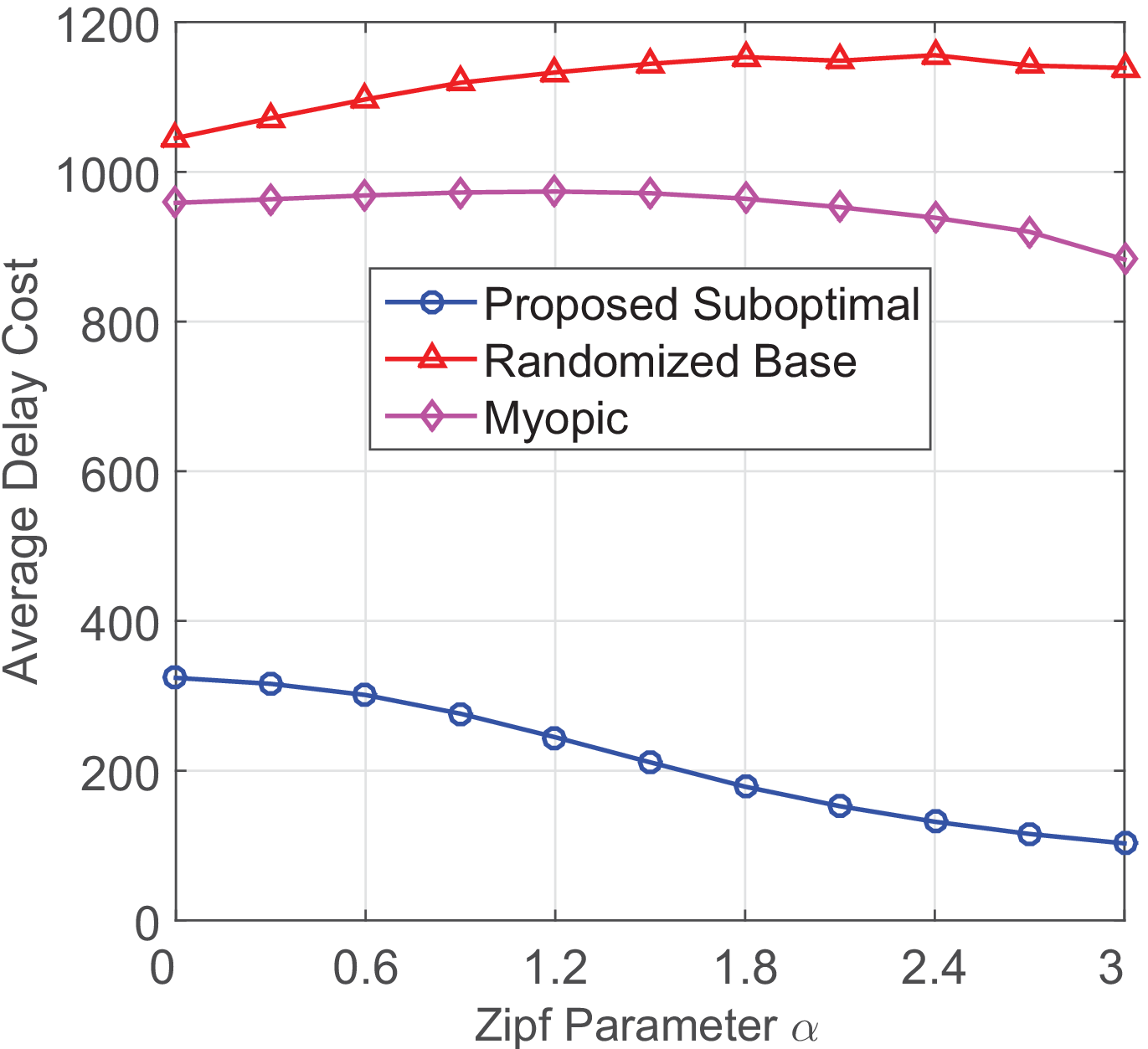}
 \subcaption{Average delay cost.}\label{fig:suboptimal_Zipf_delay}
\end{minipage}
\begin{minipage}[t]{.32\linewidth}
\centering
\includegraphics[scale=0.36]{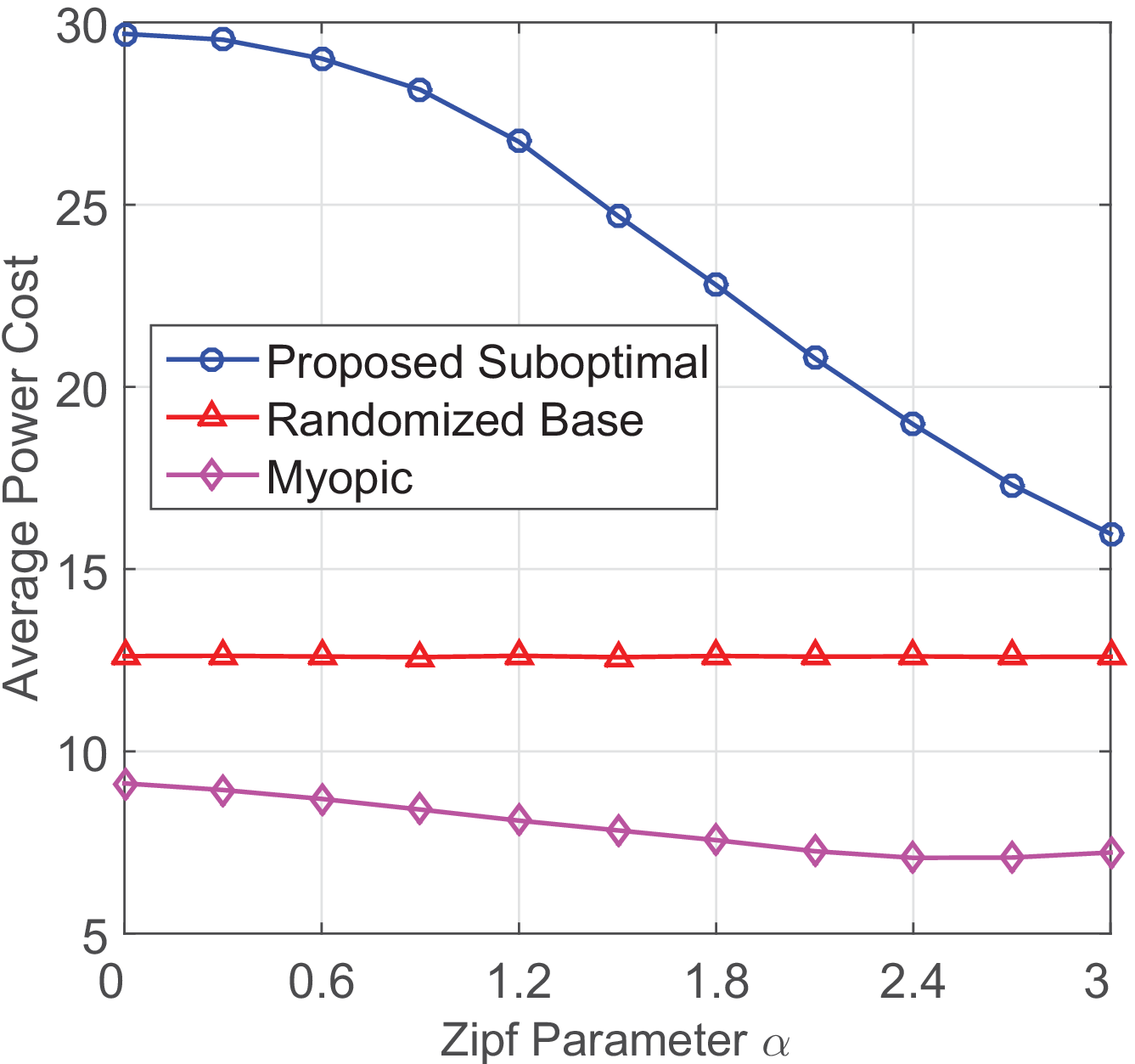}
 \subcaption{Average power cost.}\label{fig:suboptimal_Zipf_power}
\end{minipage}
\caption{Average network, delay, and power costs versus Zipf parameter $\alpha$. $K=30$ and $|\mathcal{M}_1|=|\mathcal{M}_2|=10$.}\label{fig:suboptimal_Zipf}
\end{figure*}

In Fig.~\ref{fig:suboptimal_Zipf}, Fig.~\ref{fig:suboptimal_CacheSize}, and Fig.~\ref{fig:suboptimal_NumUser}, we investigate the impacts of the Zipf parameter, the cache size of SBSs, and the number of users  on the performance of the proposed suboptimal policy and the two baseline policies, respectively.  We consider a network with 1 MBS, 2 SBSs and 20 contents. We set $w=3$, $\mathcal{M}_0 = \mathcal{M}$,  $|\mathcal{K}_0|:|\mathcal{K}_1|:|\mathcal{K}_2|=1:2:2$, $p(0,m)=30$ for all $m\in\mathcal{M}_0$, and  $p(n,m)=3$ for $n=1,2$ and $m\in\mathcal{M}_n$.
Here, $|\mathcal{K}_0|:|\mathcal{K}_1|:|\mathcal{K}_2|$ indicates the ratio of the numbers of
the users in $\mathcal{K}_0$, $\mathcal{K}_1$, and $\mathcal{K}_2$.
From  Fig.~\ref{fig:suboptimal_Zipf} -~\ref{fig:suboptimal_NumUser}, it can be seen that the proposed suboptimal policy outperforms the two baseline policies in terms of the average network cost.

Fig.~\ref{fig:suboptimal_Zipf} illustrates the average costs versus the Zipf parameter $\alpha$ for the aforementioned three policies.
The $\alpha$ parameter determines the ``skewness'' of the content popularity distribution, i.e., a large $\alpha$ indicates that a small number of contents account for the majority of content requests.
We observe that for the proposed suboptimal policy, as $\alpha$ increases, the average network cost decreases and better delay performance can be achieved with less transmission power.
This reveals that the proposed suboptimal policy can utilize caching more effectively as the content popularity distribution gets steeper.
\begin{figure*}[!t]
\begin{minipage}[t]{.32\linewidth}
\centering
\includegraphics[scale=0.36]{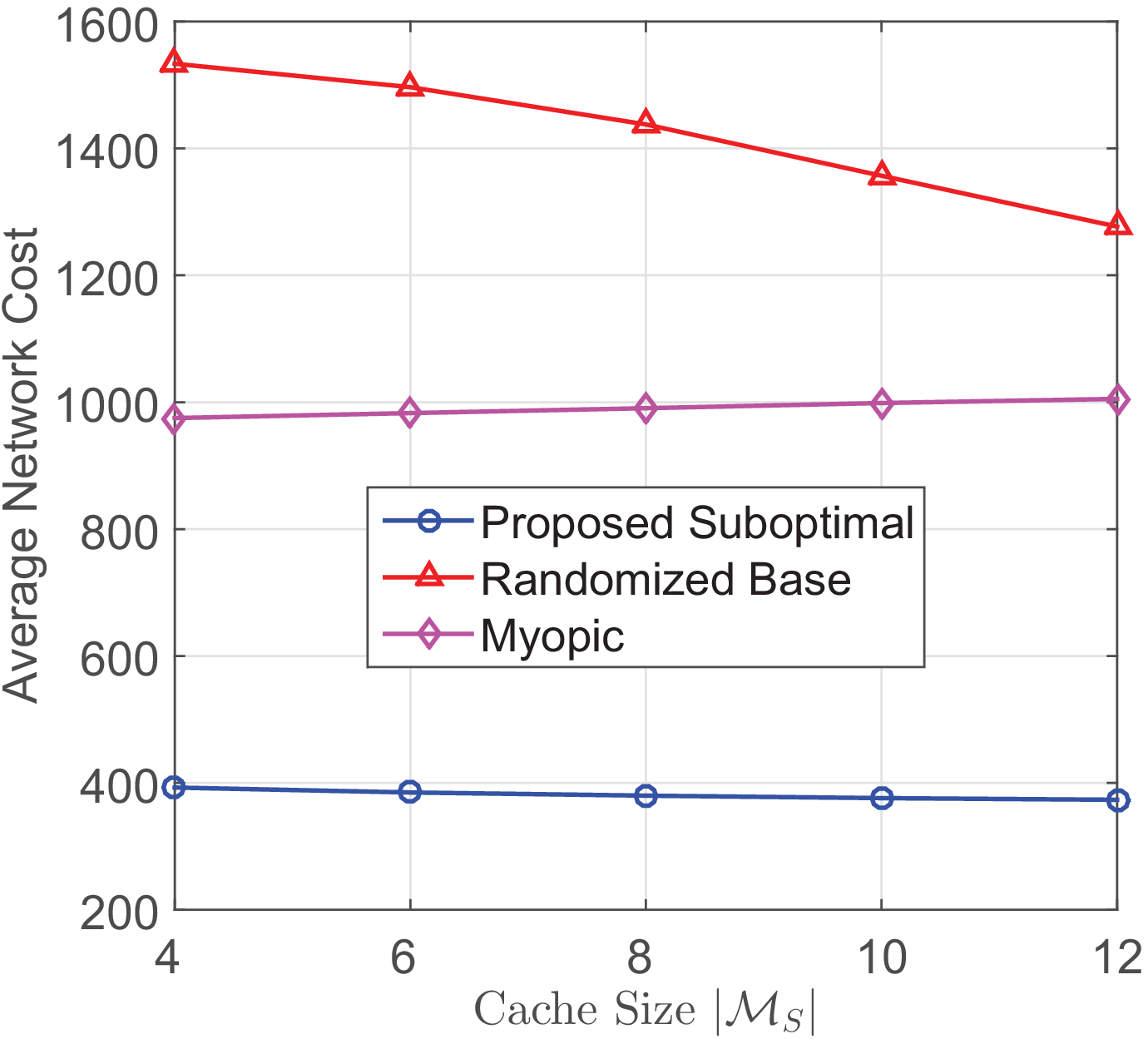}
\subcaption{Average network cost.}\label{fig:suboptimal_CacheSize_cost}
\end{minipage}
\begin{minipage}[t]{.33\linewidth}
\centering
\includegraphics[scale=0.36]{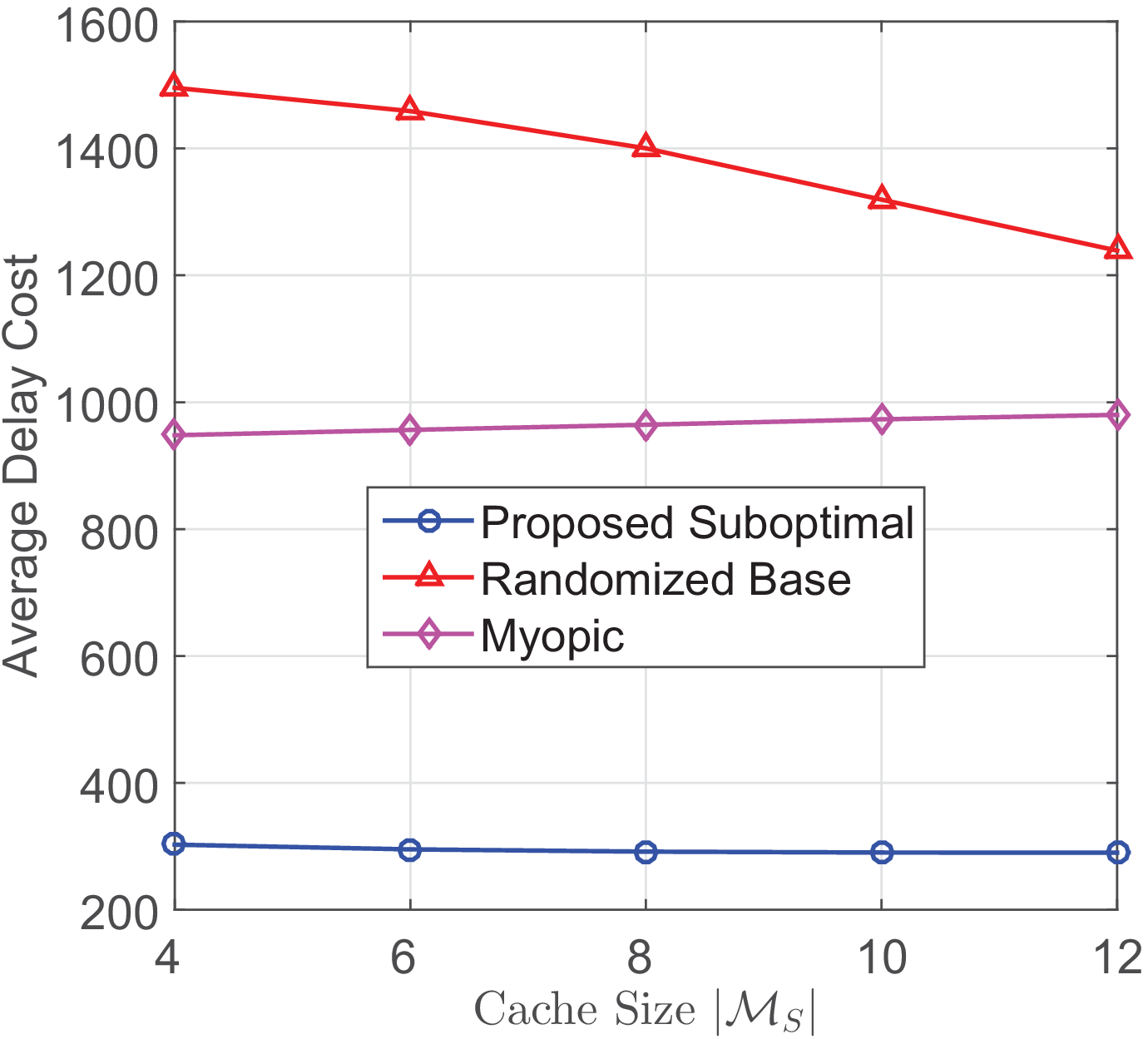}
 \subcaption{Average delay cost.}\label{fig:suboptimal_CacheSize_delay}
\end{minipage}
\begin{minipage}[t]{.32\linewidth}
\centering
\includegraphics[scale=0.36]{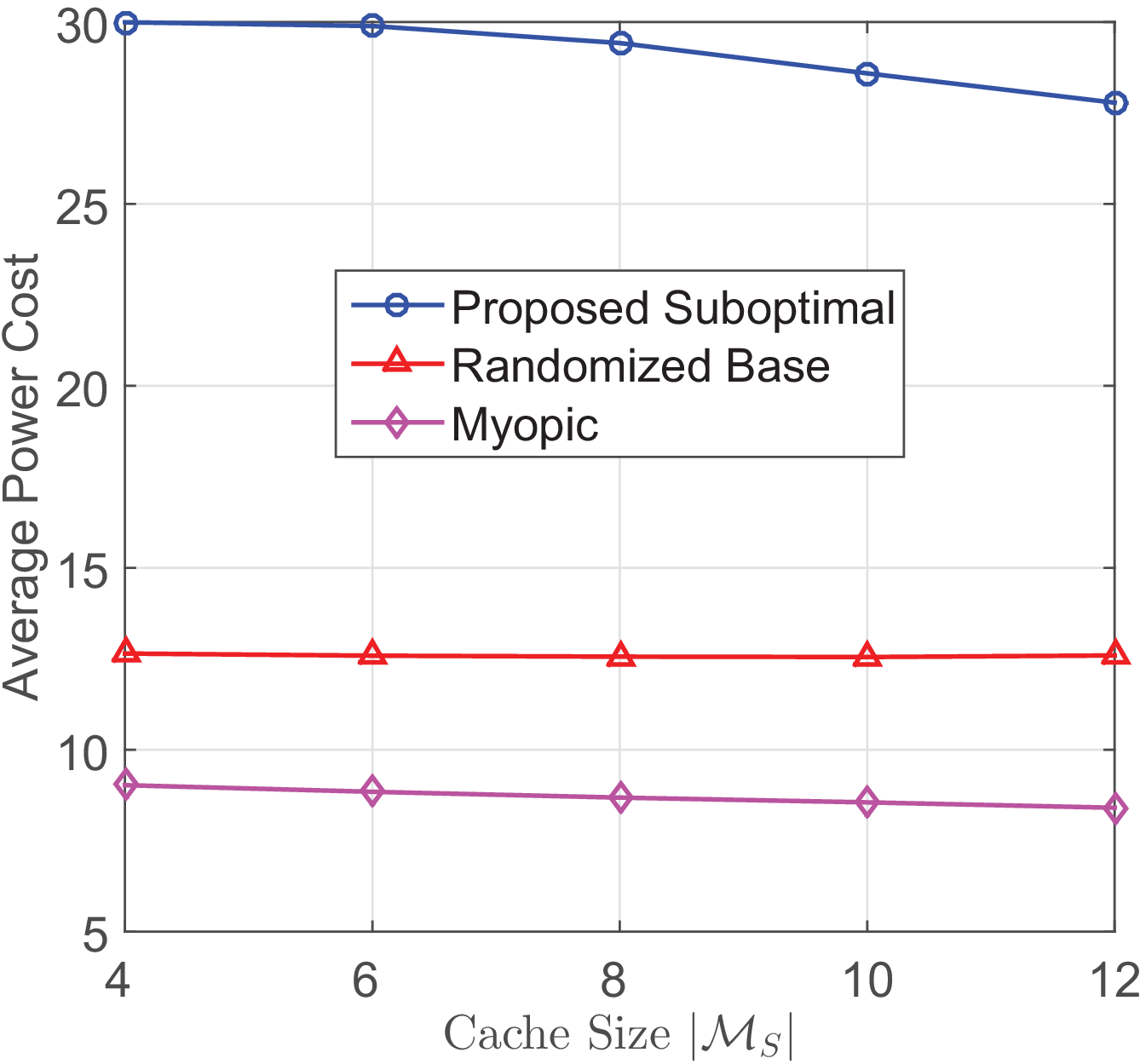}
 \subcaption{Average power cost.}\label{fig:suboptimal_CacheSize_power}
\end{minipage}
\caption{Average network, delay, and power costs versus cache size of SBSs $|\mathcal{M}_S|$. $|\mathcal{M}_S|=|\mathcal{M}_1|=|\mathcal{M}_2|$, $K=30$, and $\alpha=0.75$.}\label{fig:suboptimal_CacheSize}
\end{figure*}
\begin{figure*}[!t]
\begin{minipage}[t]{0.245\linewidth}
\centering
\includegraphics[scale=0.297]{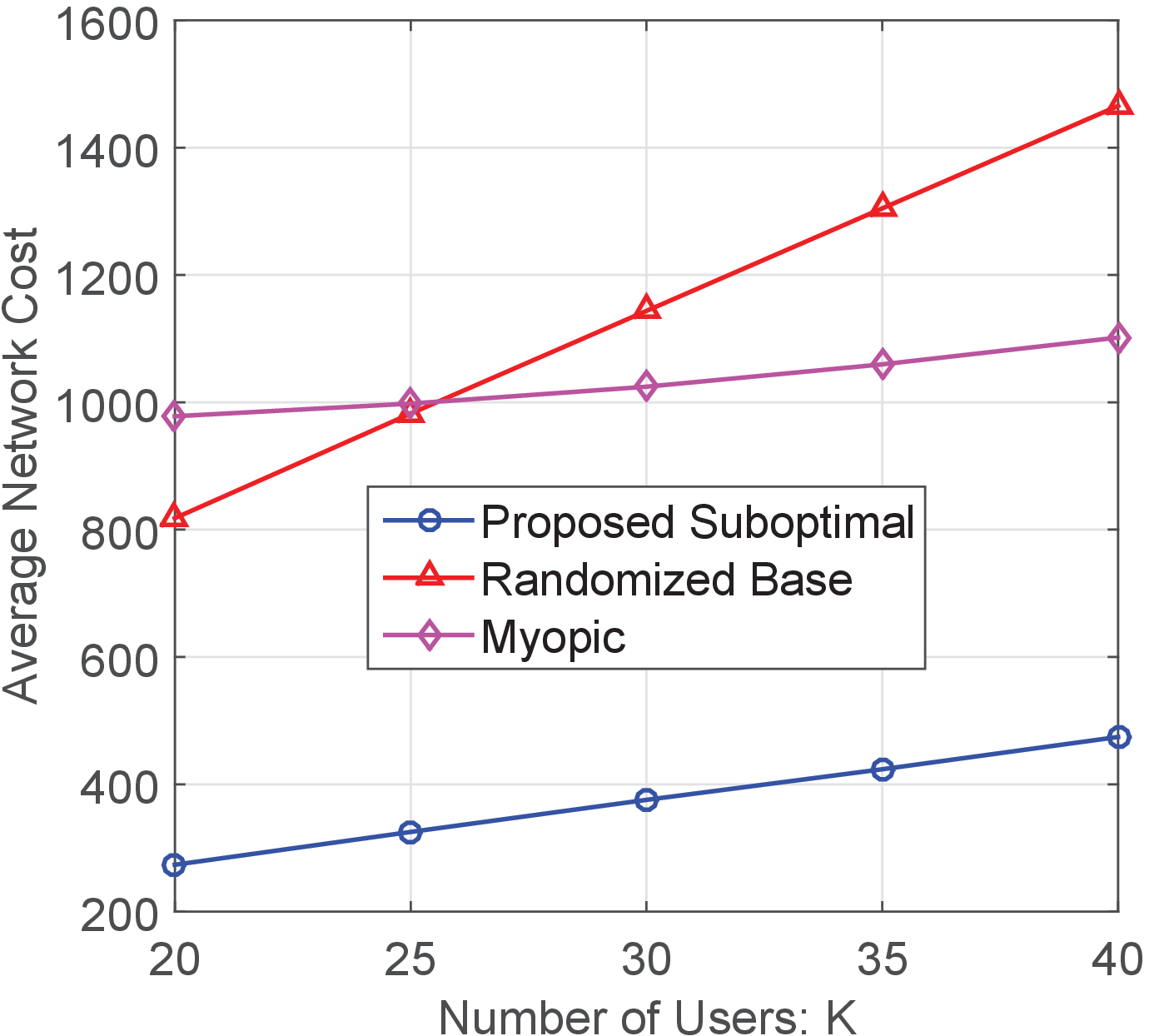}
\subcaption{Average network cost.}\label{fig:suboptimal_NumUser_cost}
\end{minipage}
\begin{minipage}[t]{.245\linewidth}
\centering
\includegraphics[scale=0.297]{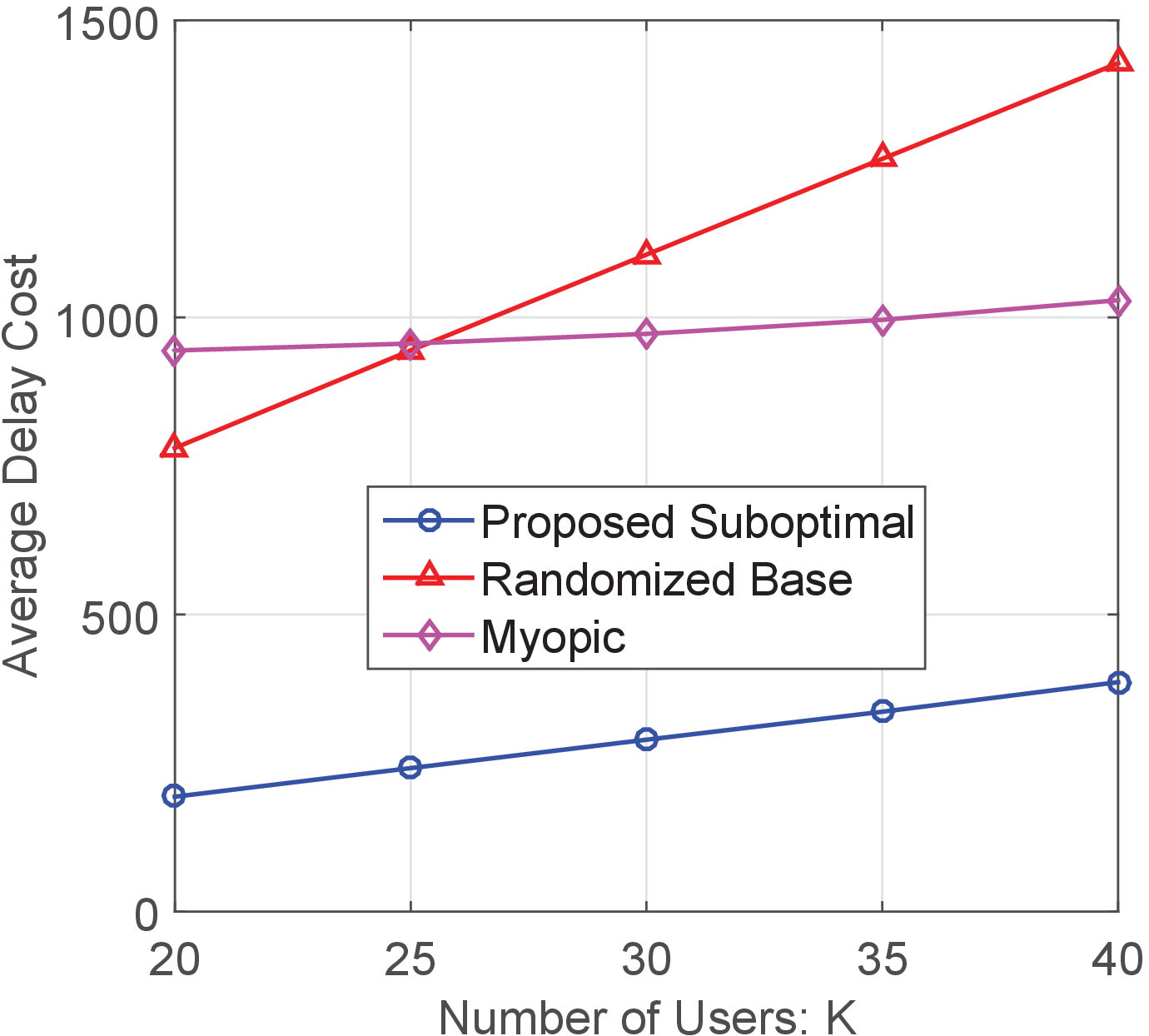}
 \subcaption{Average delay cost.}\label{fig:suboptimal_NumUser_delay}
\end{minipage}
\begin{minipage}[t]{.245\linewidth}
\centering
\includegraphics[scale=0.297]{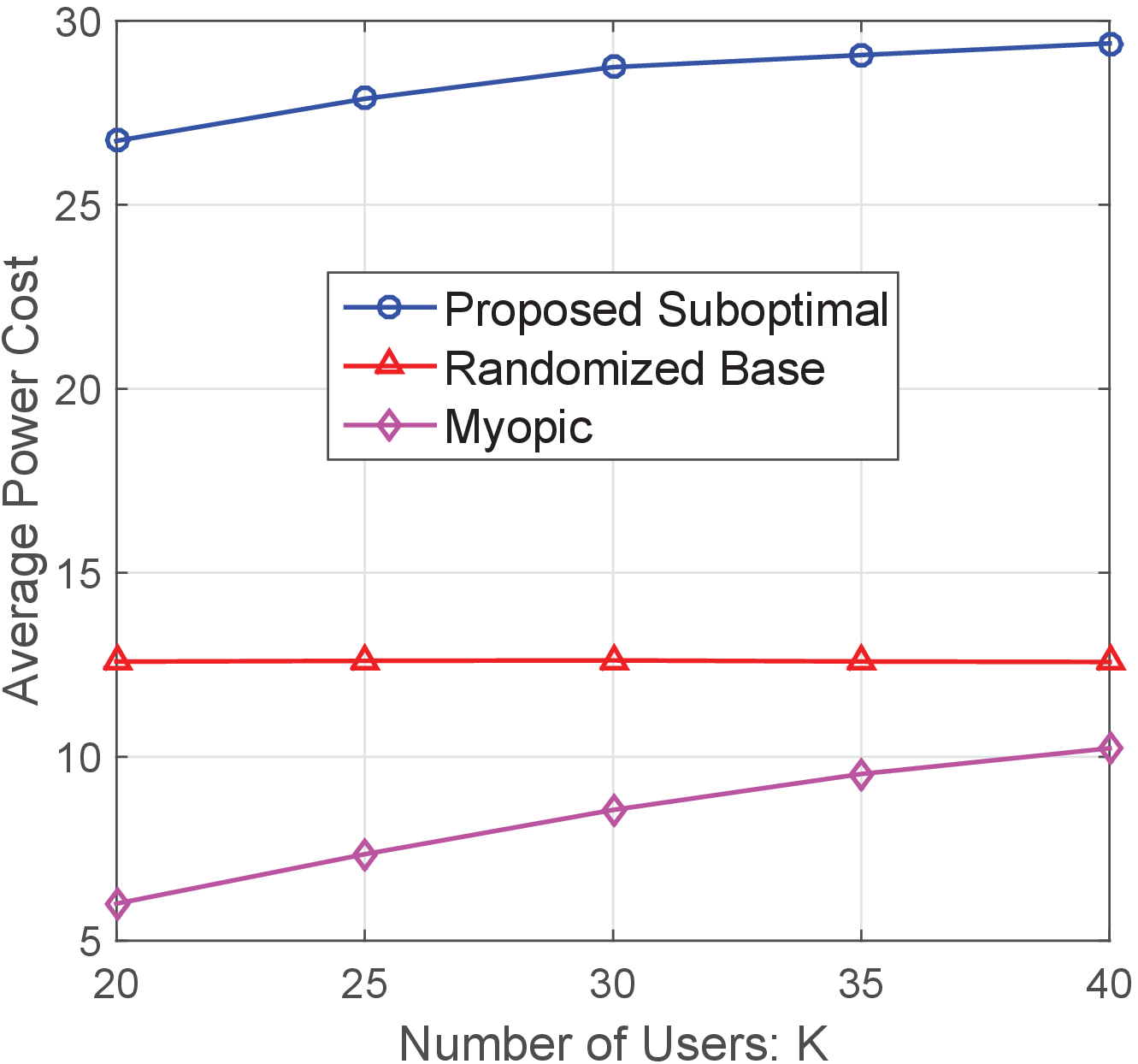}
 \subcaption{Average power cost.}\label{fig:suboptimal_NumUser_power}
\end{minipage}
\begin{minipage}[t]{.245\linewidth}
\centering
\includegraphics[scale=0.297]{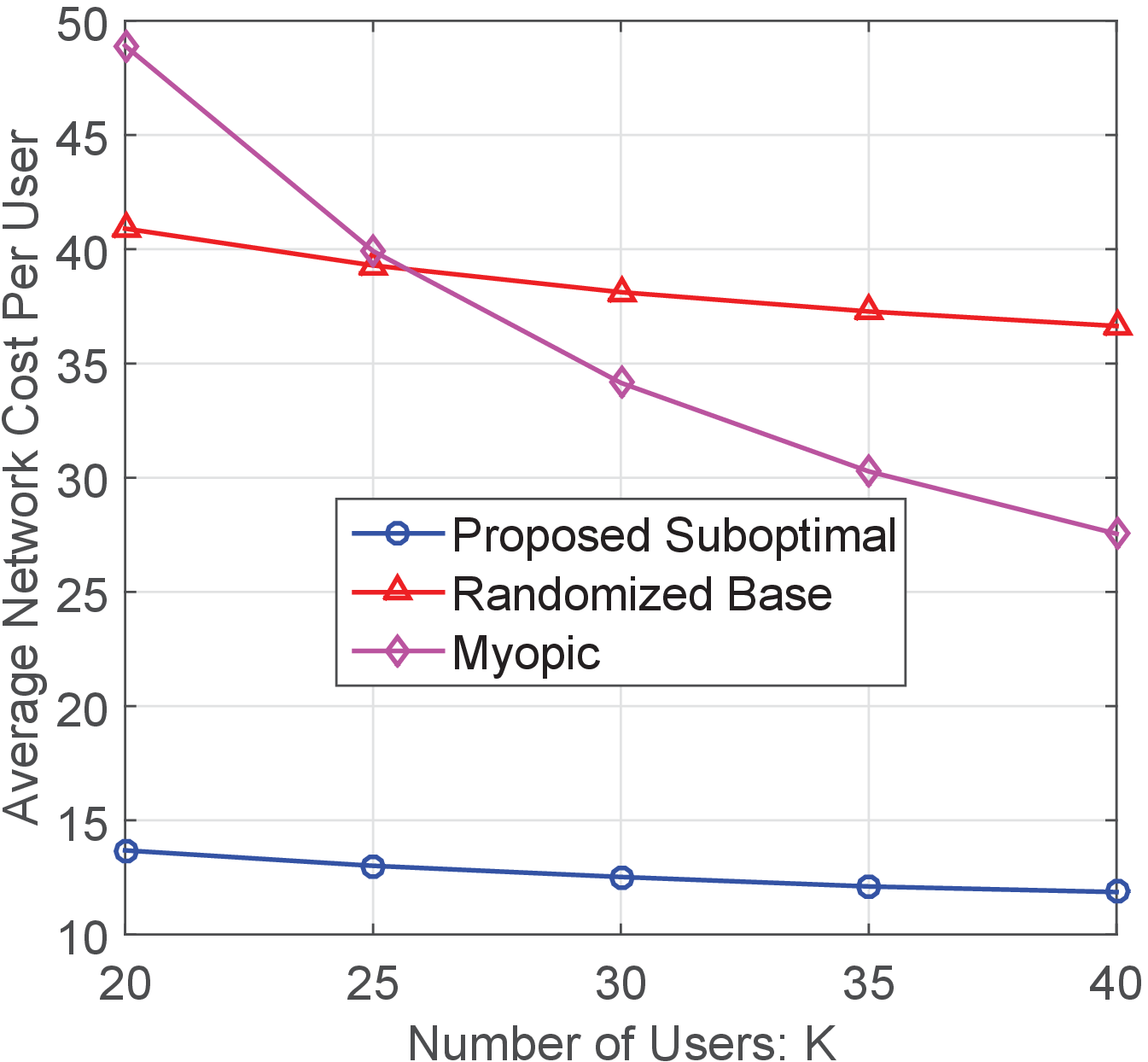}
 \subcaption{Average network cost per user.}\label{fig:suboptimal_NumUser_cost_peruser}
\end{minipage}
\caption{Average network cost, delay cost, power cost, and network cost per user versus number of users $K$. $\alpha=0.75$ and $|\mathcal{M}_1|=|\mathcal{M}_2|=10$.}\label{fig:suboptimal_NumUser}
\end{figure*}

Fig.~\ref{fig:suboptimal_CacheSize} illustrates the average costs versus the cache size of SBSs $|\mathcal{M}_S|$.
For the proposed suboptimal policy, we observe that with the increase of $|\mathcal{M}_S|$, the average network cost decreases (slightly) and the average power cost decreases without sacrificing the delay performance.
In addition, for the greedy policy, when $|\mathcal{M}_S|$ increases, the average power cost decreases, however, at the expense of the increase of the average delay and network costs.
The reason is that, when more contents are cached in SBSs, the transmission opportunities of SBSs increase and the proposed suboptimal policy can utilize these opportunities more properly than the greedy policy.
We also notice that for the greedy policy, when $|\mathcal{M}_S|$ increases, the average network and delay costs increase and the average power cost decreases. The intuitive reasons are as follows. First, the cost function  $C(\mathbf{Q},\mathbf{u})$ is not equivalent to the original objective function in Problem~\ref{problem:originalproblem}. Second, under the simulation settings, when $|\mathcal{M}_S|$ increases, i.e., the transmission opportunities of the SBSs increase, the greedy policy would be more likely to schedule the SBSs, which leads to certain reduction of the average power cost but much larger increase of the average delay cost, i.e., the increase of the average network cost. Therefore, this reveals that the greedy policy could not properly utilize the transmission opportunities offered by the increase of cache sizes.

Fig.~\ref{fig:suboptimal_NumUser} illustrates the average costs versus the number of users $K$.
From Fig.~\ref{fig:suboptimal_NumUser_cost}-\ref{fig:suboptimal_NumUser_power}, it can be seen that with the increase of $K$, the average network, delay, and power costs of the proposed suboptimal policy increase.
The reason is that, when the average request arrival rate increases (as $K$ increases), there will be more requests waiting for service and the BSs are more willing to operate instead of keeping idle.
From Fig.~\ref{fig:suboptimal_NumUser_cost_peruser}, we can see that with the increase of $K$, the average network costs per user of all policies decrease. This reflects the benefit of multicast.
\begin{table*}[!htbp]
\caption{Average computation time (sec) for different algorithms.}
\centering
\begin{tabular}{|c|c|c|c|c|}
\hline
\multicolumn{2}{|c|}{Algorithms}   & $|\mathcal{M}_0|=2, |\mathcal{M}_1|=1$  & $|\mathcal{M}_0|=3, |\mathcal{M}_1|=1$ & $|\mathcal{M}_0|=3,|\mathcal{M}_1|=2$ \\ \hline
\multirow{3}{*}{\begin{tabular}{c}PIA:\\ \cite[Chapter 8.6]{puterman}\end{tabular}} & evaluation  &  3.283 & 28.955  &  330.88 \\ \cline{2-5}
                     & improvement &  2.103 & 19.136  &  187.73 \\ \cline{2-5}
                     & total       &  5.386 & 48.091  &  518.61 \\ \hline
\multirow{3}{*}{\begin{tabular}{c}SPIA:\\ Algorithm~\ref{alg:SPIA}\end{tabular}} & evaluation  & 3.288  & 28.436  & 318.63  \\ \cline{2-5}
                     & improvement & 1.133  & 8.887  &  88.360 \\ \cline{2-5}
                     & total       & 4.421  & 37.323  &  406.99 \\ \hline
\multirow{3}{*}{\begin{tabular}{c}SSA:\\ Algorithm~\ref{alg:SSA}\end{tabular}} & step~\ref{code:hatmu}  &   0.008 & 0.015  &  0.054 \\ \cline{2-5}
                     & step~\ref{code:ssa} & 0.223  & 3.859  &  36.240 \\ \cline{2-5}
                     & total       & 0.231  & 3.874  &  36.294 \\ \hline
\end{tabular}
\label{table:complexity}

\end{table*}

Next, we compare the computational complexity of the standard optimal algorithm (PIA), the proposed structure-aware optimal algorithm (SPIA) and the proposed low-complexity suboptimal algorithm (SSA).
Table~\ref{table:complexity} illustrates the computational time comparison for a network with 1 MBS, 1 SBS and 4 users.
The simulation is carried out on a Windows x86 machine with a dual-core 2.93 GHz Intel processor and 4GB RAM using Python 3.
We set $|\mathcal{K}_0|=|\mathcal{K}_1|=2$, $w=1$, $\alpha=0.75$, $N_{n,m}=4$ for all $n\in\mathcal{N}$ and $m\in\mathcal{M}_n$.
It can be seen that SPIA has much lower computational complexity than PIA, with a reduction of about 50\% in the policy improvement step and a reduction of about 20\% in the total algorithm.
Moreover, we observe that the computation times of PIA and SPIA grow rapidly with the cardinality of the system state space ($|\mathbf{\mathbf{Q}}|=\prod_{n\in\mathcal{N}}\prod_{m\in\mathcal{M}_n}|\mathcal{Q}_{n,m}| = 5^{|\mathcal{M}_0|+|\mathcal{M}_1|}$),
while the computation time of SSA grows almost linearly with $|\mathbf{\mathbf{Q}}|$.
The computational savings of SSA compared with PIA and SPIA are significant.
These verify the discussions in Section \ref{sec:SPIA} and Section \ref{sec:suboptimal}.

\section{Conclusion}
In this paper, we study optimal content delivery strategy in a  cache-enabled HetNet by taking into account the inherent multicast capability of wireless medium.
We establish a content-centric request queue model and then formulate a stochastic content multicast scheduling problem to jointly minimize the average network delay and power costs under a multiple access constraint.
This stochastic optimization problem is an infinite horizon average cost MDP.
We show that the optimal multicast scheduling policy, which is adaptive to the request
queue state, is of threshold type.
Then, we propose a structure-aware optimal algorithm to obtain the optimal policy by exploiting its structural properties.
To further reduce the complexity, we propose a low-complexity suboptimal policy, which has similar
structural properties to the optimal policy, and develop a low-complexity algorithm to obtain this policy.

This work opens up several directions for future research. First, this work focuses on multicast scheduling for given cache placement. It would be interesting to investigate the joint optimal design of content delivery and cache placement/replacement. Second, in this work, IRM is used to model the content request arrivals. It is of particular importance to consider more practical request traffic models. Finally, in this work, we assume that the coverage areas of the SBSs are disjoint and the SBSs can operate concurrently without mutual interference.  It is also of interest to take into account the interference management in cache-enabled HetNets under  more general topology and interference models.

\appendices
\section*{Appendix A: Proof of Lemma~\ref{lemma:bellman}}\label{app:bellman}
By Proposition 4.2.5 in \cite{bertsekas}, the Weak Accessibly (WA) condition holds for unichain policies. Thus, by Proposition 4.2.3 in \cite{bertsekas}, the optimal average network cost of the MDP in Problem~\ref{problem:originalproblem} is the same for all initial states. In addition, by Proposition 4.2.1 in \cite{bertsekas}, we know that the solution $(\theta,\{V(\mathbf{Q})\})$ to the following Bellman equation exists.
\begin{align}\label{eqn:original_bellman}
  &\theta+V(\mathbf{Q})=\min_{\mathbf{u}\in\mathbf{\mathcal{U}}}\Bigg\{g(\mathbf{Q},\mathbf{u})+\sum_{\mathbf{Q}'\in\mathbf{\mathcal{Q}}}\Pr[\mathbf{Q}'|\mathbf{Q},\mathbf{u}]V(\mathbf{Q}')\Bigg\},\nonumber\\
&  \hspace{65mm}\forall \mathbf{Q}\in\mathbf{\mathcal{Q}}.
\end{align}
By substituting the transition probability $\Pr[\mathbf{Q}'|\mathbf{Q},\mathbf{u}]$ given in \eqref{eqn:tranb} into \eqref{eqn:original_bellman}, we have \eqref{eqn:bellman},  which completes the proof.\QEDA

\section*{Appendix B: Proof of Lemma~\ref{lemma:propertyV}}\label{app:propertyV}
We prove Lemma~\ref{lemma:propertyV} using RVIA and mathematical induction.

First, we introduce RVIA\cite[Chapter 4.3]{bertsekas}. For each state $\mathbf{Q}\in\mathbf{\mathcal{Q}}$, let $V_l(\mathbf{Q})$ be the value function in the $l$-th iteration, where $l=0,1,\cdots$.
Define the state-action cost function in the $l$-th iteration as
\begin{align}
&J_{l+1}(\mathbf{Q},\mathbf{u})\triangleq g(\mathbf{Q},\mathbf{u})+\mathbb{E}[V_l(\mathbf{Q}')],\label{eqn:jl}
\end{align}
where $g(\mathbf{Q},\mathbf{u})$ and $\mathbf{Q}'$ are given in \eqref{eqn:system_cost} and Lemma~\ref{lemma:bellman}, respectively.
Note that $J_{l+1}(\mathbf{Q},\mathbf{u})$ is related to the R.H.S. of the Bellman equation in \eqref{eqn:bellman}.
For each $\mathbf{Q}$, RVIA calculates $V_{l+1}(\mathbf{Q})$ according to
\begin{equation}\label{eqn:RVIA}
  V_{l+1}(\mathbf{Q})=\min_{\mathbf{u}\in\mathbf{\mathcal{U}}} J_{l+1}(\mathbf{Q},\mathbf{u})-\min_{\mathbf{u}\in\mathbf{\mathcal{U}}} J_{l+1}(\mathbf{Q}^\dag,\mathbf{u}),~\forall l,
\end{equation}
 where  $\mathbf{Q}^\dag\in \mathbf{\mathcal{Q}}$ is some fixed state. Under any initialization of $V_0(\mathbf{Q})$, the generated sequence $\{V_l(\mathbf{Q})\}$ converges to $V(\mathbf{Q})$\cite[Proposition 4.3.2]{bertsekas}, i.e.,
 \begin{equation}
   \lim_{l\to\infty}V_l(\mathbf{Q})=V(\mathbf{Q}),~\forall \mathbf{Q}\in\mathbf{\mathcal{Q}},\label{eqn:converge}
 \end{equation}
 where $V(\mathbf{Q})$ satisfies the Bellman equation in \eqref{eqn:bellman}.
Let $\mu^*_l(\mathbf{Q})$  denote the control that attains the minimum of the first term in \eqref{eqn:RVIA} in the $l$-th iteration for all $\mathbf{Q}$, i.e.,
\begin{equation}
  \mu^*_l(\mathbf{Q})=\arg\min_{\mathbf{u}\in\mathbf{\mathcal{U}}}J_{l+1}(\mathbf{Q},\mathbf{u}),~~\forall \mathbf{Q}\in\mathbf{\mathcal{Q}}.\label{eqn:optimal_l}
\end{equation}
Define $\mu^*_l(\mathbf{Q})\triangleq(\mu^*_{l,n}(\mathbf{Q}))_{n\in\mathcal{N}}$, where $\mu^*_{l,n}(\mathbf{Q})$ denotes the control action of BS $n$ for state $\mathbf{Q}$.
We refer to $\mu^*_n$ as the optimal policy for the $l$-th iteration.

Next, we prove Lemma~\ref{lemma:propertyV} through induction using RVIA.
Denote $\mathbf{Q}^1\triangleq(Q^1_{n,m})_{n\in\mathcal{N},m\in\mathcal{M}_n}$ and $\mathbf{Q}^2\triangleq(Q^2_{n,m})_{n\in\mathcal{N},m\in\mathcal{M}_n}$. To prove Lemma~\ref{lemma:propertyV}, it is sufficient to show that for any $\mathbf{Q}^1,\mathbf{Q}^2 \in\mathbf{\mathbf{Q}}$ such that $\mathbf{Q}^2\succeq\mathbf{Q}^1$,
\begin{equation}
V_l(\mathbf{Q}^2)\geq V_l(\mathbf{Q}^1),\label{eqn:vl}
\end{equation}
holds for all $l=0,1,\cdots$.
Note that, $\mathbf{Q}^1$ and $\mathbf{Q}^2$ can lie on the boundaries of the state space $\mathbf{\mathcal{Q}}$.

First, we initialize $V_0(\mathbf{Q})=0$ for all $\mathbf{Q}\in\mathbf{\mathcal{Q}}$. Thus, we have $V_0(\mathbf{Q}^1)=V_0(\mathbf{Q}^2)=0$, i.e., \eqref{eqn:vl} holds for $l=0$.
Assume that \eqref{eqn:vl} holds for  some $l>0$. We will prove that \eqref{eqn:vl} also holds for $l+1$. By \eqref{eqn:RVIA}, we have
\begin{align}
&V_{l+1}(\mathbf{Q}^1)=J_{l+1}\left(\mathbf{Q}^1,\mu^*_l(\mathbf{Q}^1)\right)-J_{l+1}(\mathbf{Q}^\dag,\mu^*_l(\mathbf{Q}^\dag))\nonumber\\
&\overset{(a)}{\leq}J_{l+1}\left(\mathbf{Q}^1,\mu^*_l(\mathbf{Q}^2)\right)-J_{l+1}(\mathbf{Q}^\dag,\mu^*_l(\mathbf{Q}^\dag))\nonumber\\
&\overset{(b)}{=}\mathbb{E}[V_l(\mathbf{Q}^{1'})]+d(\mathbf{Q}^1)+wp(\mu^*_l(\mathbf{Q}^2))-J_{l+1}(\mathbf{Q}^\dag,\mu^*_l(\mathbf{Q}^\dag)),
\label{eqn:vl1q1}
\end{align}
where  $(a)$ follows from the optimality of $\mu_l^*(\mathbf{Q}^1)$ for $\mathbf{Q}^1$ in the $l$-th iteration, $(b)$ directly follows from \eqref{eqn:jl}, and $\mathbf{Q}^{1'}=(Q^{1'}_{n,m})_{n\in\mathcal{N},m\in\mathcal{M}_n}$ with $Q^{1'}_{0,m}\triangleq\min\{\mathbf{1}(\mu^*_{l,0}(\mathbf{Q}^2)\neq m)Q^{1}_{0,m}+\tilde{A}_{0,m},N_{0,m}\}$ for all $m\in\mathcal{M}_0$ and $Q^{1'}_{n,m}\triangleq\min\{\mathbf{1}(\mu^*_{l,0}(\mathbf{Q}^2)\neq m~\&~\mu^*_{l,n}(\mathbf{Q}^2)\neq m)Q^{1}_{n,m}+A_{n,m},N_{n,m}\}$ for all $n\in\mathcal{N}^+$ and $m\in\mathcal{M}_n$.
By \eqref{eqn:jl} and \eqref{eqn:RVIA}, we also have
\begin{align}
&V_{l+1}(\mathbf{Q}^2)=J_{l+1}\left(\mathbf{Q}^2,\mu^*_l(\mathbf{Q}^2)\right)-J_{l+1}(\mathbf{Q}^\dag,\mu^*_l(\mathbf{Q}^\dag))\nonumber\\
&=\mathbb{E}[V_l(\mathbf{Q}^{2'})]+d(\mathbf{Q}^2)+wp(\mu^*_l(\mathbf{Q}^2))-J_{l+1}(\mathbf{Q}^\dag,\mu^*_l(\mathbf{Q}^\dag)),
\label{eqn:vl1q2}
\end{align}
where $\mathbf{Q}^{2'}=(Q^{2'}_{n,m})_{n\in\mathcal{N},m\in\mathcal{M}_n}$ with $Q^{2'}_{0,m}\triangleq\min\{\mathbf{1}(\mu^*_{l,0}(\mathbf{Q}^2)\neq m)Q^{2}_{0,m}+\tilde{A}_{0,m},N_{0,m}\}$ for all $m\in\mathcal{M}_0$ and $Q^{2'}_{n,m}\triangleq\min\{\mathbf{1}(\mu^*_{l,0}(\mathbf{Q}^2)\neq m~\&~\mu^*_{l,n}(\mathbf{Q}^2)\neq m)Q^{2}_{n,m}+A_{n,m},N_{n,m}\}$ for all $n\in\mathcal{N}^+$ and $m\in\mathcal{M}_n$.
Then, we compare \eqref{eqn:vl1q1} and \eqref{eqn:vl1q2} term by term.

Due to $\mathbf{Q}^2\succeq\mathbf{Q}^1$, we have $\mathbf{Q}^{2'}\succeq\mathbf{Q}^{1'}$, implying that $\mathbf{E}[V_l(\mathbf{Q}^{2'})]\geq \mathbf{E}[V_l(\mathbf{Q}^{1'})]$ by the induction hypothesis.
In addition, since $d(\mathbf{Q})$ is a monotonically non-decreasing function of $\mathbf{Q}$, we have $d(\mathbf{Q}^2)\geq d(\mathbf{Q}^1)$.
Thus, we have $V_{l+1}(\mathbf{Q}^2)\geq V_{l+1}(\mathbf{Q}^1)$, i.e., \eqref{eqn:vl} holds for $l+1$. Therefore, by induction, we can show that \eqref{eqn:vl} holds for any $l$.
 By taking limits on both sides of \eqref{eqn:vl} and by \eqref{eqn:converge}, we complete the proof of Lemma~\ref{lemma:propertyV}.\QEDA

\section*{Appendix C: Proof of Lemma~\ref{lemma:property_delta}}\label{app:property_delta}
First, we derive the general relationship between $\Delta_{\mathbf{u},\mathbf{v}}(\mathbf{Q}^1)$ and $\Delta_{\mathbf{u},\mathbf{v}}(\mathbf{Q}^2)$ for any $\mathbf{u},\mathbf{v}\in\mathbf{\mathcal{U}}$ and any $\mathbf{Q}^1, \mathbf{Q}^2\in\mathbf{\mathcal{Q}}$,
which will be used to prove the two properties in Lemma~\ref{lemma:property_delta}.
By \eqref{eqn:delta_func}, for any $\mathbf{u},\mathbf{v}\in\mathbf{\mathcal{U}}$ and any $\mathbf{Q}^1, \mathbf{Q}^2\in\mathbf{\mathcal{Q}}$, we have
\begin{align}\label{eqn:prooflemma3}
 &\Delta_{\mathbf{u},\mathbf{v}}(\mathbf{Q}^1)-\Delta_{\mathbf{u},\mathbf{v}}(\mathbf{Q}^2)\nonumber\\
=&\Big(d(\mathbf{Q}^1)+wp(\mathbf{u})+\mathbb{E}[V(\mathbf{Q}^{1,\mathbf{u}})] - d(\mathbf{Q}^1)-wp(\mathbf{v})\nonumber\\&-\mathbb{E}[V(\mathbf{Q}^{1,\mathbf{v}})]\Big)
-\Big(d(\mathbf{Q}^2)+wp(\mathbf{u})+\mathbb{E}[V(\mathbf{Q}^{2,\mathbf{u}})] \nonumber\\&- d(\mathbf{Q}^2)-wp(\mathbf{v})-\mathbb{E}[V(\mathbf{Q}^{2,\mathbf{v}})]\Big)\nonumber\\
=&\mathbb{E}[V(\mathbf{Q}^{1,\mathbf{u}})]-\mathbb{E}[V(\mathbf{Q}^{1,\mathbf{v}})]-\mathbb{E}[V(\mathbf{Q}^{2,\mathbf{u}})]+\mathbb{E}[V(\mathbf{Q}^{2,\mathbf{v}})],
\end{align}
where
\begin{subequations}
\begin{align}
&\left\{
   \begin{array}{ll}
       Q^{1,\mathbf{u}}_{0,m}=\min\{\mathbf{1}(u_0\neq m)Q^1_{0,m}+\tilde{A}_{0,m},N_{0,m}\},~\forall m\in\mathcal{M}_0\\
       Q^{1,\mathbf{u}}_{n,m}=\min\{\mathbf{1}(u_0\neq m~\&~u_n\neq m)Q^1_{n,m}+A_{n,m},N_{n,m}\},\\
       \hspace{50mm}\forall n\in\mathcal{N}^+,~\forall m\in\mathcal{M}_n
   \end{array}
\right.\label{eqn:q1prime}\\
&\left\{
   \begin{array}{ll}
       Q^{1,\mathbf{v}}_{0,m}=\min\{\mathbf{1}(v_0\neq m)Q^1_{0,m}+\tilde{A}_{0,m},N_{0,m}\},~\forall m\in\mathcal{M}_0\\
       Q^{1,\mathbf{v}}_{n,m}=\min\{\mathbf{1}(v_0\neq m~\&~v_n\neq m)Q^1_{n,m}+A_{n,m},N_{n,m}\},\\
       \hspace{50mm}\forall n\in\mathcal{N}^+,~\forall m\in\mathcal{M}_n
   \end{array}
\right.\label{eqn:q2prime}\\
&\left\{
   \begin{array}{ll}
      Q^{2,\mathbf{u}}_{0,m}=\min\{\mathbf{1}(u_0\neq m)Q^2_{0,m}+\tilde{A}_{0,m},N_{0,m}\},~\forall m\in\mathcal{M}_0\\
      Q^{2,\mathbf{u}}_{n,m}=\min\{\mathbf{1}(u_0\neq m~\&~u_n\neq m)Q^2_{n,m}+A_{n,m},N_{n,m}\},~\\
       \hspace{50mm}\forall n\in\mathcal{N}^+,~\forall m\in\mathcal{M}_n
   \end{array}
\right.\label{eqn:q3prime}\\
&\left\{
   \begin{array}{ll}
       Q^{2,\mathbf{v}}_{0,m}=\min\{\mathbf{1}(v_0\neq m)Q^2_{0,m}+\tilde{A}_{0,m},N_{0,m}\},~\forall m\in\mathcal{M}_0\\
       Q^{2,\mathbf{v}}_{n,m}=\min\{\mathbf{1}(v_0\neq m~\&~v_n\neq m)Q^2_{n,m}+A_{n,m},N_{n,m}\},\\
       \hspace{50mm}\forall n\in\mathcal{N}^+,~\forall m\in\mathcal{M}_n
   \end{array}
\right.\label{eqn:q4prime}
\end{align}
\end{subequations}

Next, based on \eqref{eqn:prooflemma3}, we prove Property 1 in Lemma~\ref{lemma:property_delta}.
Suppose action $\mathbf{u}$ satisfies that  $u_0=j\in\mathcal{M}_0$, and $\mathbf{Q}^1$ and $\mathbf{Q}^2$ satisfy the following relation:
 \begin{equation}
 \begin{cases} Q^{1}_{n,m} \geq  Q^{2}_{n,m} & \text{if}~n\in\{0\}\cup\mathcal{N}_j~\text{and}~ m=j,\\
           Q^{1}_{n,m} = Q^{2}_{n,m}  &\text{otherwise.}
\end{cases}
\end{equation}
for each $n\in\mathcal{N}$ and $m\in\mathcal{M}_n$.
By comparing \eqref{eqn:q1prime} with \eqref{eqn:q3prime}, we can see that $Q^{1,\mathbf{u}}_{n,m}=Q^{2,\mathbf{u}}_{n,m}$ for all $n\in\mathcal{N}$ and $m\in\mathcal{M}_n$, i.e., $\mathbf{Q}^{1,\mathbf{u}}=\mathbf{Q}^{2,\mathbf{u}}$. Thus, we have $\mathbb{E}[V(\mathbf{Q}^{1,\mathbf{u}})]=\mathbb{E}[V(\mathbf{Q}^{2,\mathbf{u}})]$.
By comparing \eqref{eqn:q2prime} with \eqref{eqn:q4prime}, we can see that $\mathbf{Q}^{1,\mathbf{v}}\succeq\mathbf{Q}^{2,\mathbf{v}}$. Thus, by Lemma~\ref{lemma:propertyV}, we have $\mathbb{E}[V(\mathbf{Q}^{1,\mathbf{v}})]\geq\mathbb{E}[V(\mathbf{Q}^{2,\mathbf{v}})]$. Therefore, by \eqref{eqn:prooflemma3}, we have $\Delta_{\mathbf{u},\mathbf{v}}(\mathbf{Q}^1)\leq \Delta_{\mathbf{u},\mathbf{v}}(\mathbf{Q}^2)$, which completes the proof of Property 1 in Lemma~\ref{lemma:property_delta}.

Finally, based on \eqref{eqn:prooflemma3}, we prove Property 2 in Lemma~\ref{lemma:property_delta}. Suppose action $\mathbf{u}$ satisfies that  $u_i=j\in\mathcal{M}_i$ for some $j\in\mathcal{N}^+$, and $\mathbf{Q}^1$ and $\mathbf{Q}^2$ satisfy the following relation:
 \begin{equation}
 \begin{cases} Q^{1}_{n,m} \geq  Q^{2}_{n,m} & \text{if}~n=i~\text{and}~m=j,\\
           Q^{1}_{n,m} = Q^{2}_{n,m}  &\text{otherwise.}
\end{cases}
\end{equation}
 for each $n\in\mathcal{N}$ and $m\in\mathcal{M}_n$.
  Similarly, by comparing \eqref{eqn:q1prime} with \eqref{eqn:q3prime} and comparing \eqref{eqn:q2prime} with \eqref{eqn:q4prime}, we can see that $\mathbf{Q}^{1,\mathbf{u}}=\mathbf{Q}^{2,\mathbf{u}}$ and $\mathbf{Q}^{1,\mathbf{v}}\succeq\mathbf{Q}^{2,\mathbf{v}}$. Therefore, by Lemma~\ref{lemma:propertyV} and \eqref{eqn:prooflemma3}, we have $\Delta_{\mathbf{u},\mathbf{v}}(\mathbf{Q}^1)\leq \Delta_{\mathbf{u},\mathbf{v}}(\mathbf{Q}^2)$, which completes the proof of Property 2 in Lemma~\ref{lemma:property_delta}.
\QEDA

\section*{Appendix D: Proof of Theorem~\ref{theorem:optimal}}\label{app:theorem1}
We first prove Property 1 of Theorem~\ref{theorem:optimal}.
Consider multicast scheduling action $\mathbf{u}=\mathbf{0}$, BS $n\in\mathcal{N}$, content $m\in\mathcal{M}_n$, another action $\mathbf{v}=(v_i)_{i\in\mathcal{N}}\in\mathbf{\mathcal{U}}$ where $v_n=m$, and request queue state $\mathbf{Q}=(Q_{i,j})_{i\in\mathcal{N},j\in\mathcal{M}_i}\in\mathbf{\mathcal{Q}}$ where $Q_{n,m} = \phi_{\mathbf{0}}^+(\mathbf{Q}_{-n,-m})$.
Note that, if  $\phi_{\mathbf{0}}^+(\mathbf{Q}_{-n,-m})=-\infty$, Property 1 of Theorem~\ref{theorem:optimal} always holds.
Therefore, in the following, we only consider $\phi_{\mathbf{0}}^+(\mathbf{Q}_{-n,-m})>-\infty$.
According to the definition of $\phi_{\mathbf{u}}^+(\mathbf{Q}_{-n,-m})$ in \eqref{eqn:phi}, we can see that, $\Delta_{\mathbf{u},\mathbf{v}}(\mathbf{Q})\leq 0$, i.e., $\mathbf{u}$ dominates $\mathbf{v}$ at state $\mathbf{Q}$.
Now consider another queue state $\mathbf{Q}'=(Q'_{i,j})_{i\in\mathcal{N},j\in\mathcal{M}_i}\in\mathbf{\mathcal{Q}}$ where $Q'_{n,m}\leq Q_{n,m}$ and $Q'_{i,j} = Q_{i,j}$ for all $(i,j)\neq (n,m)$. Since $\Delta_{\mathbf{u},\mathbf{v}}(\mathbf{Q})=-\Delta_{\mathbf{v},\mathbf{u}}(\mathbf{Q})$, by Lemma~\ref{lemma:property_delta}, we know that $\Delta_{\mathbf{u},\mathbf{v}}(\mathbf{Q})$ is monotonically non-decreasing with $Q_{n,m}$. Thus, we have
\begin{equation}
 \Delta_{\mathbf{u},\mathbf{v}}(\mathbf{Q}') \leq \Delta_{\mathbf{u},\mathbf{v}}(\mathbf{Q})\leq 0,
\end{equation}
i.e.,  $\mathbf{u}$ dominates $\mathbf{v}$ at state $\mathbf{Q}'$.
By the definition of $\mathbf{\mathcal{Q}}_0$, we can see that if $\mathbf{Q}\in\mathbf{\mathcal{Q}}_0$, $\mathbf{u}$ dominates all $\mathbf{v}\in\mathbf{\mathcal{U}}$ and $\mathbf{v}\neq \mathbf{u}$ at state $\mathbf{Q}$, i.e., $\mu^*(\mathbf{Q})=\mathbf{0}$. We complete the proof of Property 1 in Theorem~\ref{theorem:optimal}.

 Next, we prove Property 2 of Theorem~\ref{theorem:optimal}.
Consider BS $n\in\mathcal{N}$, content $m\in\mathcal{M}_n$, multicast scheduling action $\mathbf{u}=(u_i)_{i\in\mathcal{N}}\in\mathbf{\mathcal{U}}$ where $u_n=m$, and request queue state $\mathbf{Q}=(Q_{i,j})_{i\in\mathcal{N},j\in\mathcal{M}_i}\in\mathbf{\mathcal{Q}}$ where $Q_{n,m} = \phi_{\mathbf{u}}^-(\mathbf{Q}_{-n,-m})$.
Similar to the proof of Property 1 of Theorem~\ref{theorem:optimal}, we only need to consider that $\phi_{\mathbf{u}}^-(\mathbf{Q}_{-n,-m})<+\infty$.
According to the definition of $\phi_{\mathbf{u}}^-(\mathbf{Q}_{-n,-m})$ in \eqref{eqn:phi}, we can see that, $\Delta_{\mathbf{u},\mathbf{v}}(\mathbf{Q})\leq 0$ for all $\mathbf{v}\in\mathbf{\mathcal{U}}$ and $\mathbf{v}\neq \mathbf{u}$, i.e., $\mu^*(\mathbf{Q})=\mathbf{u}$.
\begin{itemize}
\item We first prove the first part of Property 2.
Consider queue state $\mathbf{Q}^1=(Q^1_{i,j})_{i\in\mathcal{N},j\in\mathcal{M}_i}\in\mathbf{\mathcal{Q}}$ where $Q^1_{n,m}\geq Q_{n,m}$ and $Q^1_{i,j} = Q_{i,j}$ for all $(i,j)\neq (n,m)$. To prove the first part of Property 2, it is equivalent to show that $\mu^*(\mathbf{Q}^1)=\mathbf{u}$.
By Lemma~\ref{lemma:property_delta}, for all $\mathbf{v}\in\mathbf{\mathcal{U}}$ and $\mathbf{v}\neq \mathbf{u}$, we have
\begin{equation}
 \Delta_{\mathbf{u},\mathbf{v}}(\mathbf{Q}^1) \leq \Delta_{\mathbf{u},\mathbf{v}}(\mathbf{Q})\leq 0,
\end{equation}
i.e., $\mu^*(\mathbf{Q}^1)=\mathbf{u}$. We complete the proof of the first part in Property 2.
\item
Then, we prove the second part of Property 2, i.e., the monotonically non-increasing property of $\phi_{\mathbf{u}}^-(\mathbf{Q}_{-0,-m})$ in terms of $Q_{n,m}$ for all $n\in\mathcal{N}_m$. Consider another queue state $\mathbf{Q}^2=(Q^2_{i,j})_{i\in\mathcal{N},j\in\mathcal{M}_i}\in\mathbf{\mathcal{Q}}$ where $Q^2_{i,j}\geq Q_{i,j}$ if $i\in\mathcal{N}_j$ and $j=m$, and $Q^2_{i,j} = Q_{i,j}$ otherwise. To prove the second part of Property 2, it is equivalent to show that $\mu^*(\mathbf{Q}^2)=\mathbf{u}$.
By Lemma~\ref{lemma:property_delta}, for all $\mathbf{v}\in\mathbf{\mathcal{U}}$ and $\mathbf{v}\neq \mathbf{u}$, we have
\begin{equation}
 \Delta_{\mathbf{u},\mathbf{v}}(\mathbf{Q}^2) \leq \Delta_{\mathbf{u},\mathbf{v}}(\mathbf{Q})\leq 0,
\end{equation}
i.e., $\mu^*(\mathbf{Q}^2)=\mathbf{u}$. We complete the proof of Property 2 in Theorem~\ref{theorem:optimal}.\QEDA
\end{itemize}
\section*{Appendix E: Proof of Lemma~\ref{lemma:separable}}\label{app:decomposition}
Along the line of the proof of Lemma 3 in \cite{harvest}, we prove the additive property of the value function.
Note that, we have $g(\mathbf{Q},\mathbf{u})=\sum_{n\in\mathcal{N}}\sum_{m\in\mathcal{M}_n}g_{n,m}(Q_{n,m},\mathbf{u})$. In addition, by the relation between the joint distribution and the marginal distribution, we have $\sum_{\mathbf{Q}\in\mathbf{\mathcal{Q}}}\Pr[\mathbf{Q}'|\mathbf{Q},\mathbf{u}]\allowbreak=\sum_{Q_{n,m}'\in\mathcal{Q}_{n,m}}\Pr[Q_{n,m}'|\mathbf{Q},\mathbf{u}]\allowbreak=\sum_{Q_{n,m}'\in\mathcal{Q}_{n,m}}\Pr[Q_{n,m}'|Q_{n,m},\mathbf{u}]$.
Therefore, by substituting $\hat{\theta}=\sum_{n\in\mathcal{N}}\sum_{m\in\mathcal{M}_n}\hat{\theta}_{n,m}$ and $\hat{V}(\mathbf{Q})=\sum_{n\in\mathcal{N}}\sum_{m\in\mathcal{M}_n}\hat{V}_{n,m}(Q_{n,m})$ into \eqref{eqn:randombellman}, we can see that the equality in \eqref{eqn:perrandombellman} holds, which completes the proof.\QEDA

\section*{Appendix F: Proof of Theorem~\ref{theorem:suboptimal}}\label{app:theorem3}
First, we show that for all $n\in\mathcal{N}$ and $m\in\mathcal{M}_n$, the per-BS-content value function $\hat{V}_{n,m}(Q_{n,m})$ satisfies
\begin{equation}\label{eqn:vnm}
\hat{V}_{n,m}(Q_{n,m}^2)\geq \hat{V}_{n,m}(Q_{n,m}^1),
\end{equation}
for any $Q_{n,m}^1,Q_{n,m}^2\in\mathcal{Q}_{n,m}$ such that $Q_{n,m}^2 \geq Q_{n,m}^1$.
Similar to the proof of Lemma~\ref{lemma:propertyV}, we define:
\begin{align}
&\hat{J}^{l+1}_{n,m}(Q_{n,m})\triangleq\mathbb{E}^{\hat{\mu}}\left[g_{n,m}(Q_{n,m},\mathbf{u})\right]\nonumber\\&+\sum_{Q_{n,m}'\in\mathcal{Q}_{n,m}}\mathbb{E}^{\hat{\mu}}\left[\Pr[Q_{n,m}'|Q_{n,m},\mathbf{u}]\right]\hat{V}^l_{n,m}(Q_{n,m}'),
\end{align}
where $\hat{V}^l_{n,m}(Q_{n,m})$ denotes the per-BS-content value function in the $l$-th iteration.
For each $Q_{n,m}\in\mathcal{Q}_{n,m}$, RVIA calculates $\hat{V}^{l+1}_{n,m}(Q_{n,m})$ according to:
\begin{align}
  \hat{V}^{l+1}_{n,m}(Q_{n,m})=\hat{J}^{l+1}_{n,m}(Q_{n,m})-\hat{J}^{l+1}_{n,m}(Q_{n,m}^\dag),
\end{align}
where $Q_{n,m}^\dag\in\mathcal{Q}_{n,m}$ is some fixed state. Following the proof of Lemma~\ref{lemma:propertyV}, we can prove that
for any $Q_{n,m}^1,Q_{n,m}^2\in\mathcal{Q}_{n,m}$ such that $Q_{n,m}^2 \geq Q_{n,m}^1$, $\hat{V}_{n,m}^l(Q_{n,m}^2)\geq \hat{V}_{n,m}^l(Q_{n,m}^1)$ holds for all $l=0,1,\cdots$. Thus, we can show that \eqref{eqn:vnm} holds through induction using RVIA. Then, by \eqref{eqn:vnm}, we can show that $\hat{\Delta}_{\mathbf{u},\mathbf{v}}(\mathbf{Q})$ possesses the same monotonicity properties of $\Delta_{\mathbf{u},\mathbf{v}}(\mathbf{Q})$,  following the proof of Lemma~\ref{lemma:property_delta}.
Finally, by using the monotonicity properties of $\hat{\Delta}_{\mathbf{u},\mathbf{v}}(\mathbf{Q})$, we can show the structural properties of $\hat{\mu}^*$, following the proof of Theorem~\ref{theorem:optimal}. We complete the proof.\QEDA

\bibliographystyle{IEEEtran}
\bibliography{IEEEabrv,HetNetcaching}

\begin{thebibliography}{10}
\providecommand{\url}[1]{#1}
\csname url@samestyle\endcsname
\providecommand{\newblock}{\relax}
\providecommand{\bibinfo}[2]{#2}
\providecommand{\BIBentrySTDinterwordspacing}{\spaceskip=0pt\relax}
\providecommand{\BIBentryALTinterwordstretchfactor}{4}
\providecommand{\BIBentryALTinterwordspacing}{\spaceskip=\fontdimen2\font plus
\BIBentryALTinterwordstretchfactor\fontdimen3\font minus
  \fontdimen4\font\relax}
\providecommand{\BIBforeignlanguage}[2]{{%
\expandafter\ifx\csname l@#1\endcsname\relax
\typeout{** WARNING: IEEEtran.bst: No hyphenation pattern has been}%
\typeout{** loaded for the language `#1'. Using the pattern for}%
\typeout{** the default language instead.}%
\else
\language=\csname l@#1\endcsname
\fi
#2}}
\providecommand{\BIBdecl}{\relax}
\BIBdecl

\bibitem{ccdwn}
B.~Zhou, Y.~Cui, and M.~Tao, ``Stochastic content-centric multicast scheduling
  for cache-enabled heterogeneous cellular networks,'' in \emph{ACM CoNEXT 2015
  Workshop on CCDWN}, Dec. 2015.

\bibitem{Cisco}
Cisco, ``Cisco visual networking index: Global mobile data traffic forecast
  update, 2014-2019,'' \emph{White Paper}, Feb 2015.

\bibitem{6211486}
A.~Ghosh, N.~Mangalvedhe, R.~Ratasuk, B.~Mondal, M.~Cudak, E.~Visotsky,
  T.~Thomas, J.~Andrews, P.~Xia, H.~Jo, H.~Dhillon, and T.~Novlan,
  ``Heterogeneous cellular networks: From theory to practice,'' \emph{{IEEE}
  Commun. Mag.}, vol.~50, no.~6, pp. 54--64, June 2012.

\bibitem{molisch2014caching}
A.~F. Molisch, G.~Caire, D.~Ott, J.~R. Foerster, D.~Bethanabhotla, and M.~Ji,
  ``Caching eliminates the wireless bottleneck in video aware wireless
  networks,'' \emph{Advances in Electrical Engineering}, vol. 2014, 2014.

\bibitem{Debbah}
E.~Bastug, M.~Bennis, and M.~Debbah, ``Living on the edge: The role of
  proactive caching in {5G} wireless networks,'' \emph{{IEEE} Commun. Mag.},
  vol.~52, no.~8, pp. 82--89, Aug 2014.

\bibitem{paschos2016wireless}
G.~Paschos, E.~Ba{\c{s}}tu{\u{g}}, I.~Land, G.~Caire, and M.~Debbah, ``Wireless
  caching: Technical misconceptions and business barriers,'' \emph{arXiv
  preprint arXiv:1602.00173}, 2016.

\bibitem{7282567}
M.~A. Maddah-Ali and U.~Niesen, ``Cache-aided interference channels,'' in
  \emph{Proc. IEEE ISIT}, June 2015, pp. 809--813.

\bibitem{femto}
K.~Shanmugam, N.~Golrezaei, A.~Dimakis, A.~Molisch, and G.~Caire,
  ``Femtocaching: Wireless content delivery through distributed caching
  helpers,'' \emph{{IEEE} Trans. Inf. Theory}, vol.~59, no.~12, Dec 2013.

\bibitem{7218465}
M.~Dehghan, A.~Seetharam, B.~Jiang, T.~He, T.~Salonidis, J.~Kurose, D.~Towsley,
  and R.~Sitaraman, ``On the complexity of optimal routing and content caching
  in heterogeneous networks,'' in \emph{Proc. IEEE INFOCOM}, April 2015.

\bibitem{6665021}
A.~Liu and V.~Lau, ``Cache-enabled opportunistic cooperative {MIMO} for video
  streaming in wireless systems,'' \emph{{IEEE} Trans. Signal Process.},
  vol.~62, no.~2, pp. 390--402, Jan 2014.

\bibitem{Bastug2015}
E.~Bastug, M.~Bennis, M.~Kountouris, and M.~Debbah, ``Cache-enabled small cell
  networks: modeling and tradeoffs,'' \emph{EURASIP Journal of Wireless
  Communications and Networking}, vol. 2015, p.~1, 2015.

\bibitem{7194828}
C.~Yang, Y.~Yao, Z.~Chen, and B.~Xia, ``Analysis on cache-enabled wireless
  heterogeneous networks,'' \emph{{IEEE} Trans. Wireless Commun.}, vol.~PP,
  no.~99, pp. 1--1, 2015.

\bibitem{embms}
D.~Lecompte and F.~Gabin, ``Evolved multimedia broadcast/multicast service
  ({eMBMS}) in {LTE}-advanced: overview and {Rel}-11 enhancements,''
  \emph{{IEEE} Commun. Mag.}, vol.~50, no.~11, pp. 68--74, 2012.

\bibitem{coded}
M.~Maddah-Ali and U.~Niesen, ``Fundamental limits of caching,'' \emph{{IEEE}
  Trans. Inf. Theory}, vol.~60, no.~5, pp. 2856--2867, May 2014.

\bibitem{7249208}
U.~Niesen and M.~Maddah-Ali, ``Coded caching for delay-sensitive content,'' in
  \emph{Proc. IEEE ICC}, June 2015.

\bibitem{TWC16}
K.~Poularakis, G.~Iosifidis, V.~Sourlas, and L.~Tassiulas, ``Exploiting caching
  and multicast for {5G} wireless networks,'' \emph{IEEE Trans. Wireless
  Commun.}, vol.~15, no.~4, pp. 2995--3007, April 2016.

\bibitem{ton}
N.~Abedini and S.~Shakkottai, ``Content caching and scheduling in wireless
  networks with elastic and inelastic traffic,'' \emph{{IEEE/ACM} Trans.
  Netw.}, vol.~22, no.~3, pp. 864--874, June 2014.

\bibitem{ISIT}
B.~Zhou, Y.~Cui, and M.~Tao, ``Optimal dynamic multicast scheduling for
  cache-enabled content-centric wireless networks,'' in \emph{Proc. IEEE ISIT},
  June 2015.

\bibitem{bertsekas}
D.~P. Bertsekas, \emph{Dynamic programming and optimal control, 3rd edition,
  volume II}.\hskip 1em plus 0.5em minus 0.4em\relax Belmont, MA: Athena
  Scientific, 2011.

\bibitem{6572969}
M.~Shifrin, R.~Atar, and I.~Cidon, ``Optimal scheduling in the hybrid-cloud,''
  in \emph{Proc. IFIP/IEEE International Symposium on Integrated Network
  Management (IM 2013)}, May 2013, pp. 51--59.

\bibitem{shifrin2015coded}
M.~Shifrin, A.~Cohen, O.~Gurewitz, and O.~Weisman, ``Coded retransmission in
  wireless networks via abstract {MDPs}: Theory and algorithms,'' \emph{arXiv
  preprint arXiv:1502.02893}, 2015.

\bibitem{altman2010forever}
E.~Altman, R.~El-Azouzi, D.~S. Menasche, and Y.~Xu, ``Forever young: Aging
  control for smartphones in hybrid networks,'' \emph{arXiv preprint
  arXiv:1009.4733}, 2010.

\bibitem{batch}
C.~H. Xia, G.~Michailidis, N.~Bambos, and P.~W. Glynn, ``Optimal control of
  parallel queues with batch service,'' \emph{Probability in the Engineering
  and Informational Sciences}, vol.~16, no.~03, pp. 289--307, 2002.

\bibitem{cdc}
R.~Gummadi, ``Optimal control of a broadcasting server,'' in \emph{Proc. 48th
  IEEE Conf. Decision Control (CDC/CCC)}, Dec. 2009.

\bibitem{djonin2007mimo}
D.~V. Djonin and V.~Krishnamurthy, ``{MIMO} transmission control in fading
  channels¡ªa constrained markov decision process formulation with monotone
  randomized policies,'' \emph{{IEEE} Trans. Signal Process.}, vol.~55, no.~10,
  pp. 5069--5083, 2007.

\bibitem{OR}
A.~H. Elwany, N.~Z. Gebraeel, and L.~M. Maillart, ``Structured replacement
  policies for components with complex degradation processes and dedicated
  sensors,'' \emph{Oper. Res.}, vol.~59, no.~3, pp. 684--695, 2011.

\bibitem{rosensweig2013steady}
E.~J. Rosensweig, D.~S. Menasche, and J.~Kurose, ``On the steady-state of cache
  networks,'' in \emph{Proc. IEEE INFOCOM}, April 2013.

\bibitem{deb2014multi}
K.~Deb, ``Multi-objective optimization,'' in \emph{Search methodologies}.\hskip
  1em plus 0.5em minus 0.4em\relax Springer, 2014, pp. 403--449.

\bibitem{berry2002communication}
R.~A. Berry and R.~G. Gallager, ``Communication over fading channels with delay
  constraints,'' \emph{IEEE Trans. Inf. Theory}, vol.~48, no.~5, pp.
  1135--1149, 2002.

\bibitem{levorato2009optimal}
M.~Levorato, U.~Mitra, and M.~Zorzi, ``On optimal control of wireless networks
  with multiuser detection, hybrid {ARQ} and distortion constraints,'' in
  \emph{Proc. IEEE INFOCOM}, April 2009.

\bibitem{Koole}
G.~Koole, ``Monotonicity in markov reward and decision chains: Theory and
  applications,'' \emph{Foundations and Trends in Stochastic Systems}, vol.~1,
  no.~1, pp. 1--76, 2006.

\bibitem{puterman}
M.~L. Puterman, \emph{Markov decision processes: discrete stochastic dynamic
  programming}.\hskip 1em plus 0.5em minus 0.4em\relax New York, NY, USA:
  Wiley, 2009, vol. 414.

\bibitem{harvest}
Y.~Cui, V.~Lau, and Y.~Wu, ``Delay-aware {BS} discontinuous transmission
  control and user scheduling for energy harvesting downlink coordinated {MIMO}
  systems,'' \emph{{IEEE} Trans. Signal Process.}, vol.~60, no.~7, pp.
  3786--3795, July 2012.

\bibitem{zipf}
L.~Breslau, P.~Cao, L.~Fan, G.~Phillips, and S.~Shenker, ``Web caching and
  {Zipf-like} distributions: evidence and implications,'' in \emph{Proc. IEEE
  INFOCOM}, March 1999.

\bibitem{powell2007approximate}
W.~B. Powell, \emph{Approximate Dynamic Programming: Solving the curses of
  dimensionality}.\hskip 1em plus 0.5em minus 0.4em\relax John Wiley \& Sons,
  2007, vol. 703.

\end{thebibliography}

\end{document}